\documentclass[12pt,oneside,reqno]{amsart}
\usepackage[margin=1in]{geometry}
\usepackage{color}
\usepackage{esint,amssymb}
\usepackage{graphicx}
\usepackage{MnSymbol}
\usepackage{mathtools}
\usepackage[colorlinks=true, pdfstartview=FitV, linkcolor=blue, citecolor=blue, urlcolor=blue,pagebackref=false]{hyperref}
\usepackage{microtype}
\usepackage{amsmath}
\usepackage{xifthen}
\usepackage{verbatim}
\usepackage[style=numeric]{biblatex}
\definecolor{darkgreen}{rgb}{0,0.5,0}
\definecolor{darkblue}{rgb}{0,0,0.7}
\definecolor{darkred}{rgb}{0.9,0.1,0.1}

\newtheorem{theorem}{Theorem}
\newtheorem{proposition}[theorem]{Proposition}
\newtheorem{lemma}[theorem]{Lemma}
\newtheorem{corollary}[theorem]{Corollary}

\theoremstyle{definition}
\newtheorem{remark}[theorem]{Remark}

\newtheorem{definition}[theorem]{Definition}

\newcommand{\cref}[1]{Corollary~\ref{c.#1}}

\numberwithin{equation}{section}
\numberwithin{theorem}{section}

\newcommand{\test}[1][]{%
\ifthenelse{\equal{#1}{}}{omitted}{given}%
}

\newcommand{\pa}{\partial}

\renewcommand{\part}{\partial}

\usepackage{dsfont}

\begin{document}

\title{Model knot solutions for the Twisted Bogomolny Equations}

\begin{abstract}
In this paper we prove existence for model knot solutions to the dimensionally reduced twisted Kapustin-Witten equations on $\mathbb R^3_{+}$ for any twisting parameter $t \in (0,\infty)$. We start with the explicit solutions for $t=1$ derived in \cite{WFivebranes} and perform a continuity argument in $t$. This corroborates a prediction of Gaiotto and Witten \cite[][p. 961]{GW}.
\end{abstract}

\author[P. Dimakis]{Panagiotis Dimakis}
\address[Panagiotis Dimakis]{450 Serra Mall, Bldg 380, Stanford, CA 94305}
\email{pdimakis@stanford.edu}
\subjclass[2010]{}
\date{\today}

\maketitle
\section{Introduction}

The Kapustin-Witten (KW) equations were introduced in \cite{KW} in a monumental attempt to use quantum field theory in order to understand certain aspects of the geometric Langlands correpsondence. These equations are defined on a four dimensional manifold $M^4$ equipped with a $G$-bundle $E$. If $A$ is a connection on $E$ and $\phi$ is an ad$(E)$-valued $1$-form, then the KW equations take the form 

\begin{equation}\label{KWE}
\begin{split}
F_A-\phi\wedge\phi +\star d_A\phi &= 0 \\
d_A\star\phi &= 0.
\end{split}
\end{equation}

In their paper, Kapustin and Witten derived these equations as a specialization of a $1$-parameter family of equations to a physically meaningful parameter value. The general family, which can be thought of as a phase-shifted form of complexified self-duality equations, can be obtained in the following way: If we define the complex connection $\mathcal A = A+i\phi$ and compute its curvature $F_{\mathcal A}$, this family of equations can be written as 

\begin{equation}\label{KWS}
e^{i\theta}F_{\mathcal A} = \bar{\star e^{i\theta}F_{\mathcal A}},
\end{equation}
where $\theta\in[0,\pi/2]$. The KW equations as written above, correspond to $\theta=\pi/4$. 

In subsequent papers, Witten in \cite{WFivebranes} and Gaiotto and Witten in \cite{GW} presented strong evidence that one should be able to derive knot invariants through the study of appropriate solutions of the Kapustin-Witten equations. More specifically, they developed the following conjecture:  the solution spaces of the KW equations when $M^4 = W^3\times \mathbb R^+$ , where $W \times \{0\}$ contains a knot $K$, and where we impose a certain set of singular boundary conditions along $(W \times \{0\}) \setminus K$ and separately along $K$, should contain enough information to capture the Jones polynomial when $W = S^3$ and to define a generalization of the Jones polynomial in general. 

The simplest test case for this conjecture is when $W = \mathbb R^3$ and the knot $K$ is a straight line. Witten in \cite{WFivebranes} was able to find explicit solutions to equations \eqref{KWE} with gauge group $G=SU(2)$. Witten's key observation was that by dimensionally reducing the KW equations in the direction of the line representing the knot, one ends up with a system of equations, called the Extended Bogomolny equations (EBE), which have a Hermitian Yang-Mills structure. As it will be explained below, this structure reduces the problem to solving just one equation, called the moment map equation. By requiring that the solution be rotationally and scale invariant and satisfy the aforementioned asymptotic conditions along $(\mathbb R^2\times \{0\}) \setminus \{0\}$ and seperately along $\{0\}$, Witten reduced the problem to a non-linear ordinary differential equation, which he was able to solve explicitly. These solutions which are called the model knot solutions, are indexed by a non-negative integer $k$, which for physical reasons goes by the name magnetic charge of the knot. They are the building blocks of the theory since any solution should look like one of these model solutions near a point of an arbitrary knot. 

The conjecture in the case of the straight unknot was settled in \cite{MW2} through the construction of a non linear Weitzenbock formula sensitive to the asymptotic conditions and through the study of the mapping properties of the linearization of equations \eqref{KWE}. Despite the validity of the above conjecture for the "unknot", there has been evidence, already explained in \cite[e.g.,][section 2.3]{GW}, the Kapustin-Witten equations as given in \eqref{KWE} together with the asymptotic conditions used in \cite{MW1} and \cite{MW2} do not provide sufficient information for the derivation of the Jones polynomial in the general case. 

In \cite{GW} the authors suggested two possible alternatives to overcome this issue. The first alternative, which goes by the name of complex symmetry breaking, was to require a different asymptotic behavior for the solutions studied. We will pursue ideas related to complex symmetry breaking in a companion paper to the present one. The second alternative, which is what motivated the writing of this paper, was to deform the equations \eqref{KWE} by changing the $\theta$ parameter in the family \eqref{KWS} and hope that for appropriately chosen angle $\theta$ and without changing the asymptotic conditions, the solutions would provide enough information to derive the Jones polynomial. In particular, they conjectured the existence of local model knot solutions of arbitrary charge $k$ with the same asymptotic behavior as the explicit solutions of Witten. The main theorem of this paper proves the existence of such model knot solutions:

\begin{theorem}
Given any natural number $k$ and angle $\theta\in(0,\pi/2)$, there exists a model knot solution of charge $k$ to the $\theta$-deformed KW equations \eqref{KWS}.
\end{theorem}

The idea of proof of this theorem is very similar in spirit to the deformation argument of He -Mazzeo \cite{MH2}. By dimensionally reducing the $\theta$-deformed KW equations we obtain a system of equations, which we will call the Twisted Extended Bogomolny Equations (TEBE), that continues to possess a Hermitian Yang-Mills structure. As we will explain in greater detail below, this means that there exist first order differential operators $\mathcal D_i$ ,$i~=~1,2,3$ such that the TEBE equations assume the form 

\begin{equation*}
\begin{split}
 &[\mathcal D_i , \mathcal D_j ]=0,~\text{i,j} =1,2,3, \\
\Lambda(\cos^2\beta&[\mathcal D_1, \mathcal D_1^{\dagger}]-\sin^2\beta[\mathcal D_2 , \mathcal D_2^{\dagger}])+[\mathcal D_3 , \mathcal D_3^{\dagger}] =0.
\end{split}
\end{equation*}
The parameter $\beta$ in these equations is related to the parameters $\theta$ and $t$ in the following way: 

\begin{equation*}
\tan\theta = t = \tan(\pi/4+3\beta/2).
\end{equation*}
Since the equations of the form $[\mathcal D_i , \mathcal D_j ]=0$ are invariant under complex gauge transformations, we can choose particualrly simple $\mathcal D_i$ that satisfy the first three equations and then find a particular gauge so that the final equation is satisfied as well. In this case, we reduce the problem to the study of the final equation

\begin{equation}\label{momentmap}
\Lambda(\cos^2\beta[\mathcal D_1, \mathcal D_1^{\dagger}]-\sin^2\beta[\mathcal D_2 , \mathcal D_2^{\dagger}])+[\mathcal D_3 , \mathcal D_3^{\dagger}] =0,
\end{equation}
where the unknown becomes a gauge transformation, or equivalently as we will show in the next section, a hermitian metric. At this point, our approach diverges from that of \cite{MH2}. The main reason is that in the present context it is not obvious how to construct a good approximate metric in order to perform the corresponding continuity argument. We instead start by writing the solution metric in the form \eqref{metric} and consider the ansatz \eqref{ansatz}. The motivation behind this ansatz is that a model solution of charge $k$ should possess certain symmetries and also converge to the explicit model solution of charge $k$ when $\beta \to 0$. 

The proof of the main theorem is comprized of two parts. The first part is a careful study of the linearization of \eqref{momentmap} around an exact solution for each $\beta$. We show that the linearized operator is an isomorphism between appropriately chosen Banach spaces, and apply the inverse function theorem to prove that if we already have a solution for $\beta$ then we have a solution in an open neighborhood of $\beta$. 

The second part contains the a priori estimates necessary for closedness. The key observation is that once we assume that the solution has the particular form \eqref{ansatz} mentioned above together with specific asymptotics, we can find an integrated identity that the solution has to satisfy. Using this identity, we obtain crude global bounds which we progressively refine, first in a neighborhood of zero and then everywhere else. 

As the careful reader will notice, the a-priori estimates fail precisely at $\theta=0,\pi/2$. Reading through the proof, it might seem at first that this is a mathematical artifact in the sense that maybe the ansatz \eqref{ansatz} is too restrictive in general. However, the values $0$ and $\pi/2$ are special in the following regard. If we re-write the family \eqref{KWS} as 

\begin{equation}
\begin{split}
(F_A-\phi\wedge\phi -t\star d_A\phi)^+ &= 0 \\
(F_A-\phi\wedge\phi +t^{-1}\star d_A\phi)^- &= 0 \\
d_A\star\phi &= 0
\end{split}
\end{equation}
where $t=\tan\theta$, the specialization of the equations at $t=\theta=0$ is given by 

\begin{equation}
\begin{split}
(F_A-\phi\wedge\phi)^+ &= 0 \\
(d_A\phi)^- &= 0 \\
d_A\star\phi &= 0.
\end{split}
\end{equation}
In this special value, the equations decouple. These latter equations should be considered as a four dimensional analog of the Hitchin equations. Currently we do not know whether there should exist smooth model knot solutions for $\theta = 0$ or not. We hope to address the (non-)existence of the model knot solutions as $\theta\to 0$ elsewhere. The same issue arises when $\theta\to\pi/2$.

\subsection{Acknowledgements} Before starting, I would like to thank Rafe Mazzeo for his patient guidance and insights on this paper across countless meetings and especially for his constant encouragement throughout the course of conceiving and writing this paper.

\section{The twisted Bogomolny equations}

\subsection{Deriving the equations} 

Let us begin by introducing the Twisted Kapustin Witten (TKW) equations  for a general parameter $t\in (0,\infty)$. These equations are defined just like the usual KW equations on a four dimensional manifold $M^4$ equipped with an $SU(2)$-bundle $E$. If $\hat A$ is a connection on $E$ and $\Phi$ is an ad$(E)$-valued $1$-form, then the TKW equations take the form 

\begin{equation}\label{TwistedKW}
\begin{split}
F_{\hat A} - \Phi\wedge\Phi + \frac{t-t^{-1}}{2}d_{\hat A}\Phi + \frac{t+t^{-1}}{2}\star d_{\hat A}\Phi &= 0 \\
d_{\hat A}\star\Phi &= 0.
\end{split}
\end{equation}

We are primarily interested in the case where the four dimensional manifold $M^4$ decomposes into a product as $M^4 \cong \mathbb R_{x_1}\times \Sigma_z\times\mathbb R_y^+$, with local coordinates $(x_1,z,y)$, where $\Sigma$ is a Riemann surface. We wish to dimensionally reduce the equations in the $x_1$ direction, which means that we want to consider $x_1$-invariant solutions. Expressing $\hat A= A+A_1\, dx_1$, $\Phi= \phi+\phi_1\, dx_1$ and fixing the orientation to be $\, dx_1\wedge\,dz\wedge\,dy$, the dimensionally reduced Twisted Kapustin-Witten equations take the form 

\begin{equation}\label{dimensionallyreduced}
\begin{split}
F_A - \phi\wedge\phi + \frac{t-t^{-1}}{2}d_A\phi + \frac{t+t^{-1}}{2}(\star d_A\phi_1 + \star[\phi,A_1]) &= 0\\
d_AA_1 - [\phi,\phi_1] + \frac{t-t^{-1}}{2}(d_A\phi_1 + [\phi,A_1]) - \frac{t+t^{-1}}{2}\star d_A\phi &= 0\\
d_A^{\star}\phi - [\phi_1,A_1] &= 0,
\end{split}
\end{equation}
where $\star$ is the Hodge star on $\Sigma\times\mathbb R_y^+$. In \cite{GW} Gaiotto and Witten observed that if we write $t =\tan(\pi/4 - 3\beta/2)$ and under the condition $A_1 -\tan\beta\phi_1 =0$, the above dimensionally reduced equations have a Hermitian Yang-Mills structure. In particular, if we write $d_A = D_{\bar z}\,d{\bar z} + D_z\,dz + D_y\,dy$, $\phi =\phi_{\bar z}\,d\bar z +\phi_z\,dz$ and define the operators

\begin{equation}\label{unitaryoperators}
\begin{split}
&\mathcal D_1 = (D_{\bar z} -\phi_{\bar z}\tan\beta) \,d\bar z,~ \mathcal D_2 = (D_z + \phi_z \cot\beta) \,dz,~\mathcal D_3 = D_y-i\frac{\phi_1}{\cos\beta} \\
&\mathcal D_1^{\dagger} = (D_z - \phi_z \tan\beta) \,dz,~ \mathcal D_2^{\dagger} = (D_{\bar z} +\phi_{\bar z}\cot\beta) \,d\bar z,~\mathcal D_3^{\dagger} = D_y + i\frac{\phi_1}{\cos\beta},
\end{split}
\end{equation}
the equations \eqref{dimensionallyreduced} take the following form 

\begin{equation}\label{TEBE}
\begin{split}
 &[\mathcal D_i , \mathcal D_j ]=0,~\text{i,j} =1,2,3, \\
2i\Lambda(\cos^2\beta&[\mathcal D_1, \mathcal D_1^{\dagger}]-\sin^2\beta[\mathcal D_2 , \mathcal D_2^{\dagger}])+[\mathcal D_3 , \mathcal D_3^{\dagger}] =0,
\end{split}
\end{equation}
where $\Lambda:\Omega^{1,1}\rightarrow \Omega^0$ is the inner product with the K\"{a}hler form $\frac{i}{2}\,dz\wedge\,d\bar z$. From this point onward, we will call the latter equations the Twisted Extended Bogomolny Equations (TEBE). A complete derivation of \eqref{TEBE} from \eqref{dimensionallyreduced} can be found in the appendix of \cite{MHOpers}.

\begin{remark}
As we will demonstrate in the next subsection, the operators in \eqref{unitaryoperators} are in unitary gauge.
\end{remark}

\subsection{Holomorphic versus unitary gauge} In this subsection we explain how to move between different gauge choices and derive formulas for the connection and the field coefficients. We will focus on the case where the Riemann surface $\Sigma$ is the complex plane since this is the relevant case for the present paper. See \cite{MHOpers} for a treatment of the general case. 

Fixing the bundle $E$ over $\mathbb C \times \mathbb R^+$ with hermitian metric $H$, let $\mathcal D_i$ be first order differential operators acting on sections $s\in C^{\infty}(\mathbb C \times \mathbb R^+;E)$ such that 

\begin{align*}
&\mathcal D_1 fs = (\pa_{\bar z}f) s + f\mathcal D_1s,~\mathcal D_2 fs = (\pa_z f) s + f\mathcal D_2s,~\mathcal D_3 fs = (\pa_y f) s + f\mathcal D_3s \\
&[\mathcal D_i, \mathcal D_j] =0,
\end{align*}
where $f$ is a smooth function on the base. We define the adjoints of these operators with respect to the metric $H$ in the following way:

\begin{align*}
\pa_{\bar z} H(s,s') &= H(\mathcal D_1^{\dagger} s, s') + H(s, \mathcal D_1s') \\
\pa_z H(s,s') &= H(\mathcal D_2^{\dagger} s, s') + H(s, \mathcal D_2s') \\
\pa_y H(s,s') &= H(\mathcal D_3^{\dagger} s, s') + H(s, \mathcal D_3s').
\end{align*}
Here, the notation $H(s,s'):= \bar s^TH s'$, so the inner product is complex linear in the second argument. This definition differs from the usual one but we stick to the current definition because it is consistent with most of the work done on the KW equations. We denote the set of differential operators $\{\mathcal D_1, \mathcal D_2, \mathcal D_3\}$ by $\Theta$. Next, we need to define the following operators:

\begin{equation}\label{Chernconnection}
\begin{split}
\mathcal D_z &= \sin^2\beta\mathcal D_2 + \cos^2\beta\mathcal D_1^{\dagger},~\mathcal D_{\bar z} = \sin^2\beta\mathcal D_2^{\dagger} + \cos^2\beta\mathcal D_1,\\
\phi_z &= \sin\beta\cos\beta(\mathcal D_2 - \mathcal D_1^{\dagger}),~\phi_{\bar z} = \sin\beta\cos\beta(\mathcal D_2^{\dagger} - \mathcal D_1), \\
\mathcal D_y &= \frac{\mathcal D_3 + \mathcal D_3^{\dagger}}{2},~\phi_1 = \frac{i\cos\beta}{2}(\mathcal D_3- \mathcal D_3^{\dagger}).
\end{split}
\end{equation}
From these operators we obtain a connection $\nabla_A := \mathcal D_z + \mathcal D_{\bar z} + \mathcal D_y$ and a triple $(A,\phi,\phi_1)$, called the Chern connection corresponding to the data $(E,\Theta , H)$.

We say that the operators $\Theta$ are in unitary gauge if the triple $(A,\phi,\phi_1)$ consists of unitary matrices and that they are in holomorphic gauge if $\mathcal D_1 = \pa_{\bar z}\,d\bar z$ and $\mathcal D_3 = \pa_y$. It immediately follows from this definition that the operators in \eqref{unitaryoperators} are in unitary gauge.

Working in holomorphic gauge, the commutation relations $[\mathcal D_i,\mathcal D_j] = 0$ imply that $\mathcal D_2 = \pa_z + B_z$, where $B_z$ is a trace zero matrix whose entries are holomorphic functions in $z$. As we discuss in the next subsection, in order to generate model solutions of different charges, different choices of $B_z$ will be needed. At this point it is instructive to calculate the adjoints of $\Theta$ in holomorphic gauge and explicitly write down the last equation in \eqref{TEBE}. We have that 

\begin{equation*}
\begin{split}
\pa_{\bar z} H(s,s') &= H(\pa_z s,s') + (\pa_{\bar z}H)(s,s') + H(s,\pa_{\bar z}s') \\
&= H((\pa_z +H^{-1}\pa_z H,s') + H(s,\pa_{\bar z}s')
\end{split}
\end{equation*}
and therefore $\mathcal D_1^{\dagger} = \pa_z + H^{-1}\pa_z H = H^{-1}\circ \pa_z \circ H$. Similarly, we obtain that $\mathcal D_2^{\dagger} = H^{-1}\circ (\pa_{\bar z} -B_z^{\dagger})\circ H$ and $\mathcal D_3^{\dagger} = H^{-1}\circ \pa_y \circ H$. Plugging these formulas into the last equation in \eqref{TEBE} we obtain the following equation in the metric $H$:

\begin{equation}\label{metricmoment}
4\pa_{\bar z}H^{-1}\pa_z H + 4(\pa_z + B_z)H^{-1}(\pa_{\bar z} - B_z^{\dagger})H + \pa_yH^{-1}\pa_yH =0.
\end{equation}

\begin{proposition}
Assume that the data $(E,\Theta, H)$ satisfy \eqref{TEBE}. Then the data $(E,\Theta^g, H^g)$ do as well, where $\Theta^g := (g^{-1}\circ \mathcal D_1 \circ g, g^{-1}\circ \mathcal D_2 \circ g, g^{-1}\circ \mathcal D_3 \circ g)$, and $H^g := g^{\dagger}Hg$.
\end{proposition}

\begin{proof}
The operators $\Theta^{g}$ clearly satisfy the first three equations in \eqref{TEBE}. In order to verify that the forth equation holds, we need to compute the adjoints of the modified operators with respect to the modified metric $H^{g}$. We have that 

\begin{align*}
\pa_{\bar z} H^{g}(s,s') &= H(gs,gs') \\
&= H(\mathcal D_1^{\dagger}gs,gs') + H(gs,\mathcal D_1gs')\\
&= H^{g}(g^{-1}\circ \mathcal D_1^{\dagger} \circ g s, s') + H^{g}(s, g^{-1}\circ \mathcal D_1 \circ g s'),
\end{align*}
and therefore, $\mathcal D_1^{\dagger_{H^{g}}} = g^{-1}\circ \mathcal D_1^{\dagger} \circ g$. Similarly, $\mathcal D_2^{\dagger_{H^{g}}} = g^{-1}\circ \mathcal D_2^{\dagger} \circ g$ and $\mathcal D_3^{\dagger_{H^{g}}} = g^{-1}\circ \mathcal D_3^{\dagger} \circ g$. Plugging these formulas into the last equation we obtain equation \eqref{metricmoment} conjugated with $g$ and therefore the last equation is also satisfied. 

\end{proof}

\begin{remark}
As we mentioned in the previous subsection, the TEBE equations have a Hermitian Yang-Mills structure. Let us briefly explain this now that we have developed the required vocabulary. Denote by $\mathcal G_{\mathbb C}$ the group of complex gauge transformations of $E$. As we saw, the differential operators $\mathcal D_i$ transform under such gauge transformations in the following way:

\begin{equation*}
\mathcal D_i^g = g^{-1}\circ \mathcal D_i \circ g, ~~ \mathcal D_i^{\dagger_{H^{g}}}= g^{-1}\circ \mathcal D_i^{\dagger} \circ g
\end{equation*}
It is clear that the first three equations in \eqref{TEBE} are invariant under the action of the group $\mathcal G_{\mathbb C}$, while the last equation is only invariant under the subgroup of  unitary gauge transformations $g$ with respect to the metric $H$, or more explicitly those which satisfy $g^{\dagger}Hg = H$. Because of this, the last equation in \eqref{TEBE} may be considered as the moment map for the first three equations. From here on we use this name to refer to this equation. 
\end{remark}

Finally, we show how to pass from holomorphic to unitary gauge. From the definition it follows that in unitary gauge, if $\mathcal D_1 = \pa_{\bar z} + B$ then $\mathcal D_1^{\dagger} = \pa_z - B^{\dagger}$. Writing $\mathcal D_1 = g_1^{-1} \circ \pa_{\bar z} \circ g_1$, $B = g_1^{-1}(\pa_{\bar z}g_1)$. Using the formulas above for the adjoint, $D_1^{\dagger} = g_1^{-1}H^{-1}\circ \pa_z \circ Hg_1$ and therefore, it must hold that $ g_1^{-1}(\pa_{\bar z}g_1) = g_1^{\dagger}H (\pa_z H^{-1}(g_1^{-1})^{\dagger})$. Writing $H = g^{\dagger}g$, we see that the only way this is possible is if $g_1 = g^{-1}$. Therefore, we obtain the following formulas for the operators $\Theta$ is unitary gauge:

\begin{equation}\label{unitary gauge}
\begin{split}
\mathcal D_1 &= g \circ \pa_{\bar z} \circ g^{-1} = \pa_{\bar z } - (\pa_{\bar z}g)g^{-1} \\
\mathcal D_1^{\dagger} &= (g^{-1})^{\dagger} \circ \pa_z \circ g^{\dagger} = \pa_z + (g^{-1})^{\dagger}\pa_zg^{\dagger}\\
\mathcal D_2 &= g \circ (\pa_z+ B_z) \circ g^{-1} = \pa_z - (\pa_z g)g^{-1} + gB_zg^{-1} \\
\mathcal D_2^{\dagger} &= (g^{-1})^{\dagger} \circ (\pa_{\bar z} -B_z^{\dagger}) \circ g^{\dagger} = \pa_{\bar z} + (g^{-1})^{\dagger}\pa_{\bar z}g^{\dagger} -  (g^{-1})^{\dagger} B_z^{\dagger} g^{\dagger}\\
\mathcal D_3 &= g \circ \pa_y \circ g^{-1} = \pa_y - (\pa_y g)g^{-1}\\
\mathcal D_3^{\dagger} &= (g^{-1})^{\dagger} \circ \pa_y \circ g^{\dagger} = \pa_y + (g^{-1})^{\dagger}\pa_yg^{\dagger}.
\end{split}
\end{equation}
Plugging in the expressions \eqref{unitary gauge} in the formulas \eqref{Chernconnection}, we obtain the connection and fields in unitary gauge 
\begin{equation}\label{unitary connection}
\begin{split}
A_z &= \sin^2\beta (-(\pa_zg)g^{-1} + g^{-1}B_zg) +\cos^2\beta(g^{-1})^{\dagger}\pa_zg^{\dagger} \\
\phi_z &= \cos\beta\sin\beta (-(\pa_zg)g^{-1} + g^{-1}B_zg - (g^{-1})^{\dagger}\pa_zg^{\dagger} \\
A_y &= \frac{1}{2}(-(\pa_y g)g^{-1} + (g^{-1})^{\dagger}\pa_yg^{\dagger} \\
\phi_1 &= -\frac{i}{2}\cos\beta((\pa_y g)g^{-1} + (g^{-1})^{\dagger}\pa_yg^{\dagger}),
\end{split}
\end{equation}
with $A_{\bar z} = -A_z^{\dagger}$ and $\phi_{\bar z} = -\phi_z^{\dagger}$.

\subsection{The moment map equation} In this subsection we focus on the moment map equation. In different parts of the paper we will study the moment map in different gauges. In order to prove openness it is more convenient to work in unitary gauge while the a-priori estimates for closedness will be done in holomorphic gauge. Here, we restrict to the latter and after producing an ansatz for the metric $H$, we write down the moment map equations for this ansatz. 

In order to explicitly write down the moment map equation, we need to choose $B_z$. Since we want to perturb off of solutions to the extended Bogomolny equations, it is natural to require that in the limit $\beta \to 0$, the moment map equation we obtain is the same as the one appearing in \cite[][Section 2]{MH2}. Undoubtedly, the simplest possible way of achieving this is to choose 
$B_z = \csc\beta \begin{pmatrix} 0 & z^k \\ 0 & 0 \end{pmatrix}$. With this particular choice, \eqref{metricmoment} becomes

\begin{equation}\label{explicitmoment}
\begin{split}
4\cos^2\beta~ \pa_{\bar z} H^{-1}\pa_z H & + 4\sin^2\beta ~\left(\pa_z +\csc\beta 
\begin{pmatrix} 0 & z^k \\ 0 & 0 \end{pmatrix}\right)H^{-1} 
\left(\pa_{\bar z} -\csc\beta \begin{pmatrix} 0 & 0 \\ \bar z^k & 0 \end{pmatrix}\right)  H \\
& + \pa_yH^{-1}\pa_y H =0.
\end{split}
\end{equation}
Following the notation of \cite[][Appendix C]{GW}, we express $H$ as

\begin{equation}\label{metric}
H = \begin{pmatrix} Y^{-1} & -Y^{-1}\Sigma \\ -Y^{-1}\bar\Sigma & Y + Y^{-1}|\Sigma|^2 \end{pmatrix},
\end{equation}
where $Y$ is a positive and $\Sigma$ a complex-valued function. Then 
\begin{equation*}
H^{-1} = \begin{pmatrix} Y + Y^{-1}|\Sigma|^2 & Y^{-1}\Sigma \\ Y^{-1}\bar\Sigma &Y^{-1} \end{pmatrix}.
\end{equation*}
Using these in \eqref{metricmoment} yields a system of partial differential equations:

\begin{equation}\label{momentsystem}
\begin{split}
\omega^{i}(\partial_i(Y^{-1}\bar\pa_i Y) + Y^{-2}\bar{\partial_i\Sigma}\partial_i\Sigma) + 4|z|^{2k}Y^{-2} - 2\sin\beta \, Y^{-2}(\bar z^k\pa_z\Sigma + z^k\pa_{\bar z}\bar\Sigma) &= 0 \\
\omega^{i}\partial_i(Y^{-2}\bar\partial_i\bar\Sigma) - 4\sin\beta \, \bar z^k\pa_zY^{-2} &= 0,
\end{split}
\end{equation}
where $\omega^1 = 4\cos^2\beta$, $\omega^2 = 4\sin^2\beta$ and $\omega^3 = 1$. 

As presented in \eqref{momentsystem}, this system is not easy to analyze.  Fortunately it can be simplified 
by rewriting $Y$ and $\Sigma$ in a suitable form. For the charge $k$ model solution when $\beta =0$, 
$\Sigma \equiv 0$ and $Y = yr^ke^{u(\sigma)}$, where $r=|z|$, $\sigma := y/r$ and $u$ is real-valued
and satisfies
\[
u(\sigma) \to 1~\text{as}~\sigma \to 0, \quad u(\sigma) \sim k \log \sigma ~\text{as}~\sigma\to\infty.
\]

More generally then, we shall write $Y$ in this form, for some function $u = u(\sigma)$ with the same asymptotics.
To determine a good ansatz for the form of $\Sigma$, first observe that if we insert this expression for $Y$ 
into the second equation of \eqref{momentsystem} and expand in the lowest power in $y$, we must have
\[
\Sigma \sim 2\sin\beta \, kz^{k+1}\sigma^2\ \ \mbox{as}\ \ \ y \to 0.
\]
Note also that \eqref{momentsystem} is invariant under dilations $(y,z) \to (\lambda y, \lambda z)$, $\lambda  > 0$,  
as well as under rotations $z\to e^{i\theta}z$. Finally, it is also invariant under reflections $z\to \bar z$,
$\Sigma\to\bar \Sigma$ (the latter is needed only when $\zeta^2\neq 1$).  The Nahm pole boundary  
condition with a charge $k$ knot is invariant this entire set of symmetries. This leads to the ansatz 

%
\begin{equation}\label{ansatz}
Y:= r^{k+1}e^{u(\sigma)}, \quad \Sigma:= \sin\beta \, z^{k+1}v(\sigma),
\end{equation}
with 

\begin{equation}\label{asymptotics}
\begin{split}
&u(\sigma) \sim \log \sigma~\text{as}~\sigma \to 0 \\
&u(\sigma) \sim (k+1) \log\sigma~\text{as}~\sigma\to\infty. \\
&v(\sigma) \sim 2k\sigma^2~\textrm{as}~\sigma \to 0.
\end{split}
\end{equation}
Since $\Sigma$ should be bounded along the axis where $z=0$, we also expect that 
\[
|r^{k+1}v(\sigma)| \leq C,\ \mbox{i.e.,}\ \ |v(\sigma)| \leq C \sigma^{k+1}\ \mbox{as}\ \sigma \to \infty. 
\]
Note that the supposed form of $Y$ is slightly different that the one above; the present form is simpler for computations. 
Inserting \eqref{ansatz} into \eqref{momentsystem} produces the system

\begin{equation}\label{odesystem}
\begin{split}
(\sigma^2+1)u''+ & \sigma u' \\ &
+ e^{-2u}\left[(\textrm{sin}^2\beta(2(k+1)v-\sigma v') -2)^2 + \textrm{sin}^2\beta(\cos^2\beta\sigma^2+1)v'^2\right] = 0\\
(\sigma^2+1)v''+ & \sigma v'  \\ & 
+ 4(k+1)(1-\sin^2\beta(k+1)v) - \left[ 4\sigma(1-\sin^2\beta(k+1)v) +2(\sigma^2+1)v'\right]u' = 0.
\end{split}
\end{equation}
We shall write these equations in a somewhat more tractable form in Section 4, where we analyze them more carefully.

\section{Linear Analysis}

\subsection{The linearized equation}

As the title of this subsection suggests, we start here by linearizing the moment map equation 
\begin{equation}\label{Omega}
\Omega_H := 2i\Lambda(\cos^2\beta\, [\mathcal D_1, \mathcal D_1^{\dagger}]-\sin^2\beta\, 
[\mathcal D_2 , \mathcal D_2^{\dagger}])+[\mathcal D_3 , \mathcal D_3^{\dagger}] =0.
\end{equation}
To do this, we fix a background hermitian metric $H_0$ on the bundle $E$. To linearize \eqref{Omega} around $H_0$, 
we consider variations $H = H_0e^s$ with $s\in \Gamma(\mathbb R^2\times\mathbb R^+, i\mathfrak{su}(E,H_0))$. 
Substituting this $H$ into the equation yields
\begin{proposition}
\begin{equation}\label{Omega}
\Omega_H = \Omega_{H_0} + \gamma(-s)\mathcal L_{H_0}s + Q(s),
\end{equation}
where 
\begin{equation}\label{linearization1}
\mathcal L_{H_0} s := 2i\Lambda(\cos^2\beta \, \mathcal D_1\mathcal D_1^{\dagger_{H_0}} - \sin^2\beta \,
\mathcal D_2\mathcal D_2^{\dagger_{H_0}})s + \mathcal D_3\mathcal D_3^{\dagger_{H_0}}s,
\end{equation}
$\gamma(s) := \frac{e^{\text{ad}_s} -1}{\text{ad}_s}$ and $Q(s)$ is quadratic in $s$. 
\end{proposition}

\begin{proof}
Keeping the $\mathcal D_i$ fixed, $\mathcal D_i^{\dagger_{H}} = \mathcal D_i^{\dagger_{H_0}} + 
e^{-s}(\mathcal D_i^{\dagger_{H_0}}e^s)$. Indeed, 
\begin{align*}
\pa_i H(s',s'') &= H_0(\mathcal D_i^{\dagger_{H_0}} \circ e^s s' , s'') + H_0(e^s s', \mathcal D_i s'') \\
&= H(e^{-s} \circ \mathcal D_i^{\dagger_{H_0}} \circ e^s s' , s'') + H(s',\mathcal D_i s'').
\end{align*}
Using the identity $e^{-s}(\mathcal D_i^{\dagger_{H_0}}e^s) = \gamma(-s)\mathcal D_i^{\dagger_{H_0}}s$ we indeed get 
\begin{align*}
\Omega_H = &\Omega_{H_0} + \gamma(-s)\mathcal L_{H_0}s +\\
&2i\Lambda(\cos^2\beta \, (\mathcal D_1\gamma(-s))\mathcal D_1^{\dagger_{H_0}} - \sin^2\beta \, (\mathcal D_2\gamma(-s))\mathcal D_2^{\dagger_{H_0}})s + (\mathcal D_3\gamma(-s))\mathcal D_3^{\dagger_{H_0}}s.
\end{align*}

\end{proof}

In order to state our next result, we will need to consider adjoints of operators with respect to the usual inner product on forms taking values in the bundle $E$. By this, we mean that if $A$ is an operator and $a$ and $b$ are bundle valued forms, then 

\begin{equation*}
\int \langle Aa, b \rangle_H = \int \langle a ,A^{\star}b\rangle_H.
\end{equation*}

The following lemma first appeared in \cite{MHOpers}. We include a simplified proof here for completeness. The sign diffference between the original result and the one appearing below appears because contrary to \cite{MHOpers} we stick with the unconventional definition of Hermitian adjoints throughout the paper. 

\begin{lemma}
\begin{equation}\label{linearization2}
\mathcal L_{H_0} := -\cos^2\beta~\nabla_1^{\star}\nabla_1 -\sin^2\beta~\nabla_2^{\star}\nabla_2 + \mathcal D_y^2 + \frac{\phi_1^2}{\cos^2\beta} - \frac{1}{2}[\Omega_H, \star]
\end{equation}

where $\nabla_i := \mathcal D_i + \mathcal D_i^{\dagger_{H_0}}$, $\phi_1^2:= [\phi_1, [\phi_1,\star]]$ and $\mathcal D_y$ , $\phi_1$ were defined in \eqref{Chernconnection}.
\end{lemma}

\begin{proof}

We start by recalling the Kahler identities \cite{Simpson}

\begin{equation}\label{kahlerid}
\begin{split}
i[\Lambda ,\mathcal D_1] &= -(\mathcal D_1^{\dagger_H})^{\star},~i[\Lambda, \mathcal D_1^{\dagger_H}] = \mathcal D_1^{\star}\\
i[\Lambda ,\mathcal D_2] &= (\mathcal D_2^{\dagger_H})^{\star},~i[\Lambda, \mathcal D_2^{\dagger_H}]  = -\mathcal D_2^{\star} \\
(\mathcal D_3^{\dagger_H})^{\star} &= \mathcal D_3.
\end{split}
\end{equation}

The signs are reversed from the usual Kahler identities as written down in \cite{MHOpers, Simpson} because of the unconventional definition of the Hermitian adjoints we use in this paper. Using the K\"ahler identities, we obtain
\begin{equation*}
\begin{split}
\nabla_i^{\star}\nabla_i &= \mathcal D_i^{\star}\mathcal D_i + (\mathcal D_i^{\dagger})^{\star}\mathcal D_i^{\dagger} \\
&= -i\Lambda\mathcal D_i\mathcal D_i^{\dagger} + i\Lambda\mathcal D_i^{\dagger}\mathcal D_i \\
&= -2i\Lambda \mathcal D_i\mathcal D_i^{\dagger} + i\Lambda[\mathcal D_i, \mathcal D_i^{\dagger}].
\end{split}
\end{equation*}
Then, by \eqref{Chernconnection},
\begin{equation*}
\begin{split}
4\left(\mathcal D_y^2+\frac{\phi_1^2}{\cos^2\beta}\right) &= (\mathcal D_3 + \mathcal D_3^{\dagger})^2 - (\mathcal D_3 - \mathcal D_3^{\dagger})^2 \\
&= 2(\mathcal D_3\mathcal D_3^{\dagger} + \mathcal D_3^{\dagger}\mathcal D_3)\\
&= 4\mathcal D_3\mathcal D_3^{\dagger} -2[\mathcal D_3, \mathcal D_3^{\dagger}].
\end{split}
\end{equation*}

Putting these identities together, we obtain the required formula.
\end{proof}

\begin{remark}
Notice that in the derivation of the formula for the $\nabla_i^{\star}\nabla_i $ term, it would seem that the second line is equal to $-i\Lambda[\mathcal D_i, \mathcal D_i^{\dagger}]$. However, in the Kahler identity $i[\Lambda, \mathcal D_1^{\dagger_H}] = \mathcal D_1^{\star}$, $D_1^{\dagger_H}$ denotes the hermitian adjoint acting on zero forms, which differs from the  same operator acting on $(0,1)$-forms by a sign.
\end{remark}

\subsection{Calculations in spherical coordinates} 
The aim of this subsection is to derive explicit formulas for the connection and field coefficients that appear in the linearized equation 
when written in spherical coordinates. The importance of these calculations will become apparent in the next subsection. 
From this point onward, we will be working in unitary gauge. 

Writing the metric \eqref{metric} as 
\begin{equation*}
H = g^{\dagger}g = \begin{pmatrix}Y^{-1/2} & 0 \\ -Y^{-1/2}\bar\Sigma & Y^{1/2} \end{pmatrix}\begin{pmatrix}Y^{-1/2} & -Y^{-1/2}\Sigma \\ 0 & Y^{1/2} \end{pmatrix},
\end{equation*}
and using the identities \eqref{unitary gauge}  and \eqref{unitary connection}, we obtain the following explicit formulas for the operators and the connection and field coefficients in unitary gauge respectively:
\begin{equation} \label{explicit unitary operators}
\begin{split}
\mathcal D_1 = g\circ \pa_{\bar z}\circ g^{-1} &= \pa_{\bar z} + \frac{1}{2}\begin{pmatrix}Y^{-1}\pa_{\bar z}Y &2Y^{-1}\pa_{\bar z}\Sigma \\ 0 & -Y^{-1}\pa_{\bar z}Y \end{pmatrix} \\
\mathcal D_1^{\dagger} = (g^{-1})^{\dagger}\circ \pa_z \circ g^{\dagger} &= \pa_z - \frac{1}{2}\begin{pmatrix}Y^{-1}\pa_z Y & 0 \\ 2Y^{-1}\pa_z \bar\Sigma & -Y^{-1}\pa_z Y \end{pmatrix}\\
\mathcal D_2 = g\circ \left( \pa_z + \csc\beta\begin{pmatrix} 0 & z^k \\ 0 & 0 \end{pmatrix}\right)\circ g^{-1} &= \pa_z + \frac{1}{2}\begin{pmatrix}Y^{-1}\pa_z Y & 2Y^{-1}\pa_z \Sigma + 2\csc\beta Y^{-1}z^k \\ 0 & -Y^{-1}\pa_z Y \end{pmatrix} \\
\mathcal D_2^{\dagger} = (g^{-1})^{\dagger}\circ \left( \pa_{\bar z} - \csc\beta\begin{pmatrix} 0 & 0 \\ \bar z^k & 0 \end{pmatrix}\right)\circ g^{\dagger} &= \pa_{\bar z} - \frac{1}{2}\begin{pmatrix}Y^{-1}\pa_{\bar z} Y & 0 \\ 2Y^{-1}\pa_{\bar z} \bar\Sigma + 2\csc\beta Y^{-1}\bar z^k & -Y^{-1}\pa_{\bar z} Y \end{pmatrix}\\
\mathcal D_3 = g\circ \pa_y \circ g^{-1} &= \pa_y +\frac{1}{2}\begin{pmatrix}Y^{-1}\pa_yY &2Y^{-1}\pa_y\Sigma \\ 0 & -Y^{-1}\pa_yY \end{pmatrix} \\
\mathcal D_3^{\dagger} = (g^{-1})^{\dagger}\circ \pa_y \circ g^{\dagger} &= \pa_y - \frac{1}{2}\begin{pmatrix}Y^{-1}\pa_yY & 0 \\ 2Y^{-1}\pa_y\bar\Sigma & -Y^{-1}\pa_yY \end{pmatrix}
\end{split}
\end{equation}
and 
\begin{equation} \label{explicit unitary connections}
\begin{split}
A_{\bar z} &=  \frac{\cos^2\beta}{2}\begin{pmatrix}Y^{-1}\pa_{\bar z}Y & 2Y^{-1}\pa_{\bar z}\Sigma \\ 0 & -Y^{-1}\pa_{\bar z} Y \end{pmatrix} - \frac{\sin^2\beta}{2}\begin{pmatrix}Y^{-1}\pa_{\bar z}Y & 0 \\ 2Y^{-1}\pa_{\bar z}\bar\Sigma - 2\csc\beta Y^{-1}\bar z^k  & -Y^{-1}\pa_{\bar z}Y \end{pmatrix} \\
\phi_{\bar z} &= \sin\beta\cos\beta \begin{pmatrix}-Y^{-1}\pa_{\bar z}Y & -Y^{-1}\pa_{\bar z}\Sigma \\-Y^{-1}\pa_{\bar z}\bar\Sigma + \csc\beta Y^{-1}\bar z^k  &  Y^{-1}\pa_{\bar z}Y \end{pmatrix}\\
A_z &= \frac{\sin^2\beta}{2}\begin{pmatrix}Y^{-1}\pa_zY &2Y^{-1}\pa_z\Sigma - 2\csc\beta Y^{-1}z^k  \\ 0 & -Y^{-1}\pa_zY \end{pmatrix} - \frac{\cos^2\beta}{2}\begin{pmatrix}Y^{-1}\pa_zY & 0 \\ 2Y^{-1}\pa_z \bar\Sigma & -Y^{-1}\pa_z Y \end{pmatrix} \\
\phi_z &= \sin\beta\cos\beta \begin{pmatrix}Y^{-1}\pa_zY &Y^{-1}\pa_z\Sigma- \csc\beta Y^{-1}z^k \\ Y^{-1}\pa_z\bar\Sigma & - Y^{-1}\pa_zY \end{pmatrix}\\
A_y &= \frac{1}{2}\begin{pmatrix} 0 & Y^{-1}\pa_y\Sigma \\ -Y^{-1}\pa_y\bar\Sigma & 0 \end{pmatrix}\\
\phi_1 &= \frac{i\cos\beta}{2}\begin{pmatrix}Y^{-1}\pa_yY & Y^{-1}\pa_y\Sigma \\ Y^{-1}\pa_y\bar\Sigma & -Y^{-1}\pa_yY \end{pmatrix}.
\end{split}
\end{equation}

The first step is to express the linearized operator in Euclidean coordinates. This allows us to read off the non-connection 
part of the linearized operator with relative ease. To perform this calculation, it is convenient to use the Euclidean coordinates $x_2,x_3$ 
where $z:= x_2+ix_3$. We compute each term in \eqref{linearization2} explicitly.
\begin{align*}
 \mathcal D_1 + \mathcal D_1^{\dagger} &= (\pa_{\bar z}+A_{\bar z}- \phi_{\bar z}\tan\beta) \,d\bar z +  (\pa_z+A_z- \phi_z\tan\beta) \,d z  \\
&= (\pa_2+A_2 - \phi_2\tan\beta) \,dx_2 +(\pa_3+A_3 - \phi_3\tan\beta) \,dx_3  \\
&= (D_2 - \phi_2\tan\beta) \,dx_2 +  (D_3 - \phi_3\tan\beta) \,dx_3,
\end{align*}
where $\phi_2\,dx_2 + \phi_3\,dx_3 = \phi_z\,dz + \phi_{\bar z}\,d\bar z$. The adjoint operator  $\nabla_1^{\star}$ is equal to
\begin{align*}
 \nabla_1^{\star} &= -(\pa_2+A_2 - \phi_2\tan\beta) \,\langle dx^2, \star \rangle - (\pa_3+A_3 - \phi_3\tan\beta) \,\langle dx^3, \star \rangle  \\
&= -(D_2 - \phi_2\tan\beta) (\,dx^2)^{\star} -  (D_3 - \phi_3\tan\beta) (\,dx^3)^{\star}
\end{align*}
since in unitary gauge the metric is just the usual inner product. Therefore, 
\begin{equation*}
\nabla_1^{\star}\nabla_1 = - (D_2 - \tan\beta\, \phi_2)(D_2 - \tan\beta\, \phi_2) - (D_3 - \tan\beta\, \phi_3)(D_3 - \tan\beta\, \phi_3).
\end{equation*}
Similarly one obtains the equality 
\[
\nabla_2^{\star}\nabla_2 = -  (D_2 + \cot\beta\, \phi_2)(D_2 + \cot\beta \, \phi_2) - (D_3 + \cot\beta\, \phi_3)(D_3 + \cot\beta \, \phi_3).
\]
Plugging these formulae into the linearized operator and expanding, we obtain
\begin{equation*}
\mathcal L_H = D_2^2 + D_3^2 + D_y^2 + \frac{[\phi_1,[\phi_1,\star]]}{\cos^2\beta}+ [\phi_2,[\phi_2,\star]] + [\phi_3,[\phi_3,\star]].
\end{equation*}

Now define
\begin{equation*}
\frac{N}{\rho^2}:= \frac{[\phi_1,[\phi_1,\star]]}{\cos^2\beta}+ [\phi_2,[\phi_2,\star]] + [\phi_3,[\phi_3,\star]]. 
\end{equation*}
Using \eqref{ansatz}, the operator $N$ is independent of $\rho$. In fact, we show later that $N$ has a particularly 
special form.

We are now in a position to express $\mathcal L_H$ in spherical coordinates $(\rho, \psi, \theta)$, where $y = \rho \cos \psi$ and
$z = x_2 + i x_3 = \sin \psi \rho e^{i\theta}$. This takes the form
\begin{equation}\label{spherical linearization}
\mathcal L_H = D_{\rho}^2+\frac{2}{\rho}D_{\rho} +\frac{1}{\rho^2}(D_{\psi}^2+\cot\psi D_{\psi} + \frac{1}{\sin^2\psi}D_{\theta}^2) + N.
\end{equation}
To compute the connection coefficients $A_{\rho}$, $A_{\psi}$ and $A_{\theta}$, insert \eqref{explicit unitary operators} into 
\eqref{linearization2} to get the explicit formula
\begin{equation}
\begin{split}
\mathcal L_H = \cos^2\beta &\left(\pa_{x_1}+\frac{1}{2}\begin{pmatrix}i\sin\theta\, Y^{-1}\pa_rY & Y^{-1}\pa_{\bar z}\Sigma \\ -Y^{-1}\pa_z\bar\Sigma & -i\sin\theta\, Y^{-1}\pa_rY\end{pmatrix}\right)^2 \\ 
\cos^2\beta &\left(\pa_{x_2} - \frac{i}{2}\begin{pmatrix}\cos\theta\, Y^{-1}\pa_rY & Y^{-1}\pa_{\bar z}\Sigma \\ Y^{-1}\pa_z\bar\Sigma & -\cos\theta\, Y^{-1}\pa_rY\end{pmatrix}\right)^2 \\
 \sin^2\beta &\left(\pa_{x_1}+\frac{1}{2}\begin{pmatrix}-i\sin\theta\, Y^{-1}\pa_rY & Y^{-1}\pa_z\Sigma + \csc\beta Y^{-1}z^k\\ -Y^{-1}\pa_{\bar z}\bar\Sigma - \csc\beta Y^{-1}\bar z^k & i\sin\theta\, Y^{-1}\pa_rY\end{pmatrix}\right)^2 \\
 \sin^2\beta &\left(\pa_{x_2}+\frac{i}{2}\begin{pmatrix}\cos\theta\, Y^{-1}\pa_rY & Y^{-1}\pa_z\Sigma + \csc\beta Y^{-1}z^k\\ Y^{-1}\pa_{\bar z}\bar\Sigma + \csc\beta Y^{-1}\bar z^k & -\cos\theta\, Y^{-1}\pa_rY\end{pmatrix}\right)^2\\
 & \left(\pa_y + \frac{1}{2}\begin{pmatrix} 0 & Y^{-1}\pa_y\Sigma \\ -Y^{-1}\pa_y\bar\Sigma & 0 \end{pmatrix}\right)^2 - \frac{1}{4}  \left(\begin{pmatrix}Y^{-1}\pa_yY & Y^{-1}\pa_y\Sigma \\ Y^{-1}\pa_y\bar\Sigma & -Y^{-1}\pa_yY \end{pmatrix}\right)^2.
\end{split}
\end{equation}

Next, use the identities
\begin{equation} \label{change of variables}
\begin{split}
\pa_{x_2} &= \sin\psi\cos\theta\,\pa_{\rho} -\frac{\sin\theta\,}{\rho\sin\psi}\pa_{\theta} + \frac{\cos\psi\cos\theta}{\rho}\pa_{\psi}\\
\pa_{x_3} &= \sin\psi\sin\theta\,\pa_{\rho} +\frac{\cos\theta}{\rho\sin\psi}\pa_{\theta} + \frac{\cos\psi\sin\theta}{\rho}\pa_{\psi}\\
\pa_y &= \cos\psi\pa_{\rho} - \frac{\sin\psi}{\rho}\pa_{\psi}
\end{split}
\end{equation}
to find the coefficients of the partial derivatives $\pa_{\rho},~\pa_{\psi},~\pa_{\theta}$. 
After some calculation, we obtain that the coefficient of $\pa_{\rho}$ equals 
\begin{equation*}
\begin{pmatrix}0 & A \\ -\bar A & 0\end{pmatrix} +\frac{2}{\rho},
\end{equation*}
where
\begin{equation}
A = \frac{1}{2}\left[\sin\psi\left(\cos^2\beta e^{-i\theta}Y^{-1}\pa_{\bar z}\Sigma + \sin^2\beta e^{i\theta}(Y^{-1}\pa_z\Sigma + \csc\beta Y^{-1}z^k)\right) + \cos\psi Y^{-1}\pa_y\Sigma\right]. \label{A}
\end{equation}
Similarly, the coefficient of $\pa_{\psi}$ is
\begin{equation*}
\frac{1}{\rho}\begin{pmatrix}0 & B \\ -\bar B & 0\end{pmatrix}+\frac{\cot\psi}{\rho^2},
\end{equation*}
where
\begin{equation}
B = \frac{1}{2}\left[\cos\psi\left(\cos^2\beta e^{-i\theta}Y^{-1}\pa_{\bar z}\Sigma + \sin^2\beta e^{i\theta}(Y^{-1}\pa_z\Sigma + \csc\beta Y^{-1}z^k)\right) - \sin\psi Y^{-1}\pa_y\Sigma\right], \label{B}
\end{equation}
and the coefficient of $\pa_{\theta}$ is
\begin{equation*}
\frac{1}{\rho\sin\psi}\begin{pmatrix}C & D \\ -\bar D & -C\end{pmatrix},
\end{equation*}
where 
\begin{equation}
\begin{split}
C &= -\frac{iY^{-1}\pa_rY}{2} \\
D &= -\frac{i}{2}\left[\cos^2\beta e^{-i\theta}Y^{-1}\pa_{\bar z}\Sigma + \sin^2\beta e^{i\theta}(Y^{-1}\pa_z\Sigma + \csc\beta Y^{-1}z^k)\right]. 
\end{split}
\end{equation}

These calculations imply that 
\begin{equation}\label{spherical connections}
A_{\rho} = \begin{pmatrix}0 & A \\ -\bar A & 0\end{pmatrix}, \quad A_{\psi} = 
\rho \begin{pmatrix}0 & B \\ -\bar B & 0\end{pmatrix}, \qquad A_{\theta} = r\begin{pmatrix}C & D \\ -\bar D & -C\end{pmatrix}.
\end{equation}

\bigskip

\subsection{The ansatz for the linearized problem}
As we have discussed above, the hope is that the perturbation arguments of this section will be able to produce a solution that is in the form \eqref{ansatz}. One way to do this is to show that the linearized operator \eqref{spherical linearization} is invertible between appropriate Banach spaces, and then, since the solution we would get is unique, we could argue that it must satisfy the required symmetries. Although this is arguably the more natural approach, we employ a different one here. We require that the section $s$ in unitary gauge have the form 
\begin{equation}\label{ansatz2}
s^{U} = \begin{pmatrix} \gamma & e^{i(k+1)\theta}\delta \\ e^{-i(k+1)\theta}\delta & -\gamma \end{pmatrix},
\end{equation}
where $\gamma$ and $\delta$ are real functions of $\psi$. We use the upper $U$ notation to keep track of when we work with the section in unitary gauge. The reason we seek for a solution with this particular form is because it preserves the ansatz \eqref{ansatz}. Below, we briefly demonstrate why this is the case. We calculate
\begin{equation*}
e^{s^{U}} = \sum\limits_{n = 0}^{\infty} \frac{(\gamma^2 + \delta^2)^n}{(2n)!} \text{Id} + \sum\limits_{n = 0}^{\infty} \frac{(\gamma^2 + \delta^2)^n}{(2n+1)!}s^{U}.
\end{equation*}
In order to go back to holomorphic gauge and compute the new hermitian metric $H = H_0e^s$, we need to understand how $s$ transforms when we move from unitary to holomorphic gauge. Since unitarity is preserved, $s$ must be equal to $g^{-1} s^{U} g$ where  
\[g = \begin{pmatrix}Y^{-1/2} & -Y^{-1/2}\Sigma \\ 0 & Y^{1/2} \end{pmatrix},\]
is such that $H_0 = g^{\dagger}g$ as in section $2$. Therefore, the new metric is going to be 
\[H_0e^s = g^{\dagger}e^{s^{U}}g.\]
From this we get that 
\begin{equation*}
\begin{split}
Y &= Y_0\left( \sum\limits_{n = 0}^{\infty} \frac{(\gamma^2 + \delta^2)^n}{(2n)!} + \sum\limits_{n = 0}^{\infty} \frac{(\gamma^2 + \delta^2)^n}{(2n+1)!}\right)^{-1}\\
\Sigma &= \Sigma_0 - Y_0e^{i(k+1)\theta}\delta\sum\limits_{n = 0}^{\infty} \frac{(\gamma^2 + \delta^2)^n}{(2n+1)!}\left( \sum\limits_{n = 0}^{\infty} \frac{(\gamma^2 + \delta^2)^n}{(2n)!} + \sum\limits_{n = 0}^{\infty} \frac{(\gamma^2 + \delta^2)^n}{(2n+1)!}\right)^{-1}.
\end{split}
\end{equation*}
From these formulas it is clear that $Y$ is of the form \eqref{ansatz}. As for $\Sigma$, it is enough to notice that $Y_0 = r^{k+1}F(\psi)$ and therefore $Y_0e^{i(k+1)\theta} = z^{k+1}F(\psi)$ so this term is also of the form \eqref{ansatz}. 

\begin{remark}
This calculation, on top of proving that sections of this form preserve \eqref{ansatz}, shows the exact relation between the section $s^U$ and the functions $u$ and $v$, which are the main objects of study in the next section, thus providing a bridge between sections $3$ and $4$ of this paper.
\end{remark}

In order to be able to consider a section of this form, the first thing we need to check is that the linearized equation is algebraically consistent. By this we mean that the error term should have the same algebraic form as $\mathcal L_H s$.  Thus, a good place to start, is to calculate the error term that is produced when we plug in an exact solution for $\beta_0$ into the equations \eqref{TEBE} for $\beta_1$ close to $\beta_0$. Working in unitary gauge, we write 
\begin{align*}
\mathcal D_1^{\beta_1} &= D_1^{\beta_0} - (\tan\beta_1 - \tan\beta_0)\phi_{\bar z} \\
\mathcal D_2^{\beta_1} &= D_2^{\beta_0} - (\tan\beta_1 - \tan\beta_0)\phi_z \\
\mathcal D_3^{\beta_1} &= D_3^{\beta_0} - i(\frac{1}{\cos\beta_1} - \frac{1}{\cos\beta_0})\phi_1.
\end{align*}

We will analyze only the error coming from the $\mathcal D_1$ term since the others lead to similar expressions. We want to show that the error term has the form \eqref{ansatz2} times $\rho^{-2}$. Defining $\epsilon = \tan\beta_1 - \tan\beta_0$, we write 
\begin{equation*}
\begin{split}
[\mathcal D_1^{\beta_1}, \mathcal D_1^{\beta_1,\dagger}] &= [\mathcal D_1^{\beta_0}, \mathcal D_1^{\beta_0,\dagger}] + \epsilon\left([\mathcal D_1^{\beta_0,\dagger},\phi_{\bar z}] - [\mathcal D_1^{\beta_0},\phi_z]\right)+ \epsilon^2[\phi_{\bar z},\phi_z]\\
&= [\mathcal D_1^{\beta_0}, \mathcal D_1^{\beta_0,\dagger}] + \epsilon\left(\pa_{\bar z}\phi_z - \pa_z\phi_{\bar z} + [A_{\bar z},\phi_z] - [A_z,\phi_{\bar z}]\right) + \epsilon^2[\phi_{\bar z},\phi_z].
\end{split}
\end{equation*}

At this point, notice that the matrices $\phi_z$ and $A_z$ given in \eqref{explicit unitary connections} are of the form 
\begin{equation*}
\frac{e^{-i\theta}}{\rho}\begin{pmatrix}a & e^{i(k+1)\theta}b \\ e^{-i(k+1)\theta}c & -a \end{pmatrix},
\end{equation*}
where $a,b,c$ are real functions of $\psi$, while $\phi_{\bar z}$ and $A_{\bar z}$ are of the form 
\begin{equation*}
\frac{e^{-i\theta}}{\rho}\begin{pmatrix}a & e^{i(k+1)\theta}b \\ e^{-i(k+1)\theta}c & -a \end{pmatrix}.
\end{equation*}
Combining these expressions with the formulas $\pa_z = \frac{e^{-i\theta}}{2}\left(\pa_r - \frac{i}{r}\pa_{\theta}\right)$ and $\pa_{\bar z} = \frac{e^{i\theta}}{2}\left(\pa_r + \frac{i}{r}\pa_{\theta}\right)$, a standard calculation gives the desired result. 

Now that we know that the error term has this form and we require that $s$ have the same form, we need to show that the linearized operator preserves this form as well. The key is that the commutator of two matrices of the form \eqref{ansatz2} is of the same form. Since the connection coefficients $A_{\rho}$ and  $A_{\psi}$ are of this form, the following part of the linearized operator preserves the required form
\begin{equation*}
D_{\rho}^2+\frac{2}{\rho}D_{\rho} +\frac{1}{\rho^2}(D_{\psi}^2+\cot\psi D_{\psi}).
\end{equation*}
Next, we see that for $s$ as in \eqref{ansatz2},
\begin{align*}
D_{\theta}s  &= \begin{pmatrix} 0 & i(k+1)e^{i(k+1)\theta}y \\ -i(k+1)e^{-i(k+1)\theta}y & 0 \end{pmatrix} + A_{\theta}s \\
&= \left(\frac{1}{2}\begin{pmatrix} i(k+1) & 0 \\ 0 & -i(k+1) \end{pmatrix} +A_{\theta}\right)s\\
&=: \hat A_{\theta}s,
\end{align*}
where $\hat A_{\theta}$ is has the required form up to a factor of $i$. Therefore $D_{\theta}^2 s$ also has the required form. $\phi_1$ is $i$ times a matrix of the form \eqref{ansatz2} and therefore the only term left to check is 
\begin{equation*}
[\phi_2,[\phi_2,\star]] + [\phi_3,[\phi_3,\star]].
\end{equation*}
We can re-write this sum as
\begin{equation*}
 [\phi_2,[\phi_2,\star]] + [\phi_3,[\phi_3,\star]] = -2\left([\phi_z,[\phi_z^{\dagger},\star]]+[\phi_z^{\dagger},[\phi_z,\star]]\right)
\end{equation*}
and use the expressions appearing in \eqref{explicit unitary connections} to get the required result. Therefore, we have showed that the linearized operator should at least in principle admit a solution of the form \eqref{ansatz2}. Proving that this is indeed the case, is the subject of the next subsection. 

\bigskip

\subsection{Mapping properties} Finally, we are in position to prove the openness part for the continuity argument. The linearized operator whose mapping properties we want to study is 
\begin{equation}\label{final linearization}
\Phi:= D_{\psi}^2 + \cot\psi D_{\psi} + N +\hat A_{\rho}^2 + 2\hat A_{\rho} + \frac{1}{\sin^2\psi}\hat A_{\theta}^2,
\end{equation}
where $\hat A_{\rho}:= \rho A_{\rho}$ is independent of $\rho$ and $\hat A_{\theta}$ was defined in the previous subsection. This operator turns out to be a regular-singular operator acting on sections of the form \eqref{ansatz2}. Therefore, in order to understand its mapping properties we need to calculate its indicial roots at the singular regions and introduce appropriate Banach spaces adapted to the degeneracy of this operator. 

We start with the calculation of the indicial roots of \eqref{final linearization}. We should expect singular behavior at $\psi = 0$ and $\psi =: \pi/2 -\omega = \pi/2$. This differs from the study of the indicial behavior of similar operators as in \cite{MW2} where they only have indicial behavior at $\omega = 0$. In order to produce solutions of the form \eqref{ansatz}, we considered $s$ of the form \eqref{ansatz2} and as a result, we traded the dependence in $\theta$ with the appearance of indicial behavior at $\psi = 0$. 

\begin{lemma}
The set of indicial roots of $\Phi$ at $\omega = 0$ is $\{-1,2\}$, in accordance with the Nahm pole boundary condition, and the set of indicial roots at $\psi = 0$ is $\{-k-1,0,0,k+1\}$. 
\end{lemma}

\begin{remark}
The reason that there are four indicial roots at $\psi = 0$ instead of two as one would expect from a second order equation, is because we work with a matrix-valued function $s$ and therefore we are actually dealing with a system of equations. There are only two indicial roots at $\omega = 0$ because the system asymptotically decouples in that region. 
\end{remark}

\begin{proof}

We start with the indicial behavior at $\omega = 0$. Plugging the asymptotics \eqref{asymptotics} into \eqref{spherical connections} we find that the indicial operator is given by 
\begin{equation*}
\pa_{\omega}^2 - \frac{2\text{Id}}{\omega^2}
\end{equation*}
and therefore the indicial roots are $-1$ and $2$. This calculation should come as no surprise. These are the indicial roots that characterize the Nahm pole solutions for $SU(2)$ away from the knot and therefore, they should not depend on $\beta$. 

Next, we study the indicial behavior at $\psi = 0$. In this region, the indicial operator takes the form 
\begin{equation*}
\pa_{\psi}^2 + \frac{1}{\psi}\pa_{\psi} - \frac{(k+1)^2\text{Pr}_{\text{Ant}}}{\psi^2},
\end{equation*}
where $\text{Pr}_{\text{Ant}}$ is projection to the anti-diagonal 
\begin{equation*}
\text{Pr}_{\text{Ant}}\begin{pmatrix} x & y \\ \bar y & -x \end{pmatrix} = \begin{pmatrix} 0 & y \\ \bar y & 0 \end{pmatrix}.
\end{equation*}
In order to see why this is the case, notice that as $\psi \to 0$ the asymptotics \eqref{asymptotics} imply that all of the matrices \eqref{spherical connections} and all of the fields $\phi_i$ are regular. Therefore, the part of the operator that contributes to the indicial behavior is 
\[ \pa_{\psi}^2 + \cot \psi \pa_{\psi} + \frac{1}{4}\begin{pmatrix} i(k+1) & 0 \\ 0 & -i(k+1) \end{pmatrix}^2.\]
The last term should be interpreted as the operation of conjugating with the matrix 
\[\frac{1}{2}\begin{pmatrix} i(k+1) & 0 \\ 0 & -i(k+1) \end{pmatrix}\]
twice. Keeping only the lowest terms in $\psi$ we indeed get the required formula for the indicial operator. As a result, the indicial roots are $\pm (k+1)$ and $0$, the latter with double multiplicity. 

\end{proof}

Before we state the main theorem of this section, let us define the appropriate function spaces that will appear below. These are designed to be sensitive to the singular behavior of the linearized operator as we will see.
\begin{definition}
We will work with spaces of the form 
\begin{equation*}
\hat\Gamma(S_+^2) := \Gamma_{\text{inv}}(S_+^2, i\mathfrak{su}(E,H)|_{S_+^2}).
\end{equation*}
The $\Gamma$ refers to the function space in which the norm of the section is going to belong, the $\text{inv}$ subscript means that we consider sections of the form \eqref{ansatz2} and $ i\mathfrak{su}(E,H)|_{S_+^2}$ is the restriction of the bundle on the half sphere. The first type of spaces we want to consider is the function spaces of the form
\begin{equation*}
\hat H^k_{0}(S_+^2),
\end{equation*}
where $f\in \hat H^k_{0}(S_+^2)$ if $\omega\psi\pa_{\psi}f \in \hat H^{k-1}_{0}(S_+^2)$, with $\hat H^0_{0}(S_+^2):= \hat L^2(S_+^2)$. The operator $\omega\psi\pa_{\psi}$ is a zero operator at both ends in the sense of \cite{Melliptic}. The second type of spaces we will need is 
\begin{equation*}
\hat C^{k,a}_0(S_+^2),
\end{equation*}
so if $u$ is in such a space, then
\begin{equation*}
\|u\|_{L^{\infty}} + \sup \limits_{i\le k} [(\omega\psi\pa_{\psi})^i u]_{0;0,a} < \infty 
\end{equation*}
where
\begin{equation*}
[u]_{0;0,a} = \sup \limits_{\psi\neq \psi'}\frac{|u(\psi) - u(\psi')|(\psi+\psi')^a}{|\psi-\psi'|^a}.
\end{equation*}
\end{definition}

\begin{theorem}
There is a unique self-adjoint realization 
\begin{equation*}
\Phi :\hat L^2(S_+^2) \rightarrow \hat L^2(S_+^2)
\end{equation*}
which has domain $\omega^2\psi^2\hat H^2_0(S_+^2)$. The operators 
\begin{equation}\label{invertibility}
\begin{split}
\Phi: \omega^{\gamma}\psi^{\delta}\hat H^2_0(S_+^2) &\rightarrow \omega^{\gamma-2}\psi^{\delta-2}\hat L^2(S_+^2)\\
\Phi: \omega^{\mu}\psi^{\nu}\hat C^{2,a}_0(S_+^2) &\rightarrow \omega^{\mu-2}\psi^{\nu-2}\hat C^{0,a}_0(S_+^2)
\end{split}
\end{equation}
are invertible provided that $-1/2 < \gamma < 5/2$, $1/2 < \delta < k+3/2$ and $-1 < \mu < 2$, $0 < \nu < k+1$ respectively.
\end{theorem}

\begin{proof}
 The fact that the operator $\Phi$ is essentially self-adjoint can be checked manually by an integration by parts argument. The fact that its domain in this case is $\omega^2\psi^2\hat H^2_0(S_+^2)$ follows from standard techniques developed in \cite{Melliptic}. The crucial property to check is that the kernel of $\Phi$ mapping between these function spaces is zero. Then, since the operator is self-adjoint, the index and as a result the co-kernel too is zero and the operator is invertible. Applying the results of \cite{Melliptic} give the invertibility of $\Phi$ between the more general spaces as given in \eqref{invertibility}. Assume on the contrary that the kernel is not empty. Let $s$ be an element of the kernel. Then, from the calculation of the indicial roots above, such a section must vanish at least quadratically near $\omega = 0$ and at least as fast as $\psi^{k+1}$ near $\psi = 0$. Therefore, we can perform an integration by parts since the boundary terms vanish. Remembering that we work in unitary gauge, and that all the fields are anti-hermitian, the term $\langle s , \hat A_{\rho}s\rangle$ vanishes and then 
\begin{equation*}
\begin{split}
0 &= \langle s , \Phi s\rangle \\
&= - |\sin\psi D_{\psi} s|^2 - |\Phi_1s|^2 - |\Phi_2s|^2 - |\Phi_3s|^2 - |\hat A_{\rho} s|^2 - |\frac{\hat A_{\theta}}{\sin\psi}s|^2,
\end{split}
\end{equation*}
where $\Phi_1 := \frac{\rho \phi_1}{\cos\beta}$, $\Phi_2 := \rho \phi_2$ and $\Phi_3 := \rho \phi_3$. In particular this implies that $D_{\psi} s = 0$. Since $s = 0$ for $\psi = 0$, this implies that $s$ is identically zero. Therefore there is no kernel and we are done.
\end{proof}

\subsection{Existence of solution}

Following the development of the linear analysis in the last subsection, it is now a simple matter to prove existence. Recall that from equation \eqref{Omega} we can write $\Omega$ as  

\[\Omega_H = \Omega_{H_0} + \gamma(-s)\mathcal L_{H_0}s + Q(s),\] 

where 

\[Q_{H_0}(s) = 2i\Lambda(\cos^2\beta (\mathcal D_1\gamma(-s))\mathcal D_1^{\dagger_{H_0}} - \sin^2\beta (\mathcal D_2\gamma(-s))\mathcal D_2^{\dagger_{H_0}})s + (\mathcal D_3\gamma(-s))\mathcal D_3^{\dagger_{H_0}}s.\]

If we denote by $\mathcal G_{H_0}$ the inverse of $\mathcal L_{H_0}$, then to rephrase our problem as a fixed point problem 
\[ 
s = \mathcal G_{H_0}( \Omega_{H_0} + Q_{H_0}(s)),
\] 
it is enough to prove a bound of the form 
\[ \|\omega^{2-\mu}\psi^{2-\nu}(Q(s_1) - Q(s_2))\|_{0;0,a} \le C(\|\omega^{-\mu}\psi^{-\nu}s_1\|_{2;0,a} + \|\omega^{-\mu}\psi^{-\nu}s_1\|_{2;0,a})\|\omega^{-\mu}\psi^{-\nu}(s_1-s_2)\|_{2;0,a}\]
as long as $\|\omega^{-\mu}\psi^{-\nu}s_1\|_{2;0,a} + \|\omega^{-\mu}\psi^{-\nu}s_1\|_{2;0,a} \le \epsilon$ for some appropriately small $\epsilon >0$. At this point, we re-write 
\[ 
Q_{H_0}(s) = 2i\Lambda(\cos^2\beta (\delta(-s)\mathcal D_1 s)\gamma(-s)\mathcal D_1^{\dagger_{H_0}}s - \sin^2\beta (\delta(-s)\mathcal D_2 s)\gamma(-s)\mathcal D_2^{\dagger_{H_0}}s) + (\delta(-s)\mathcal D_3 s)\gamma(-s)\mathcal D_3^{\dagger_{H_0}}s,
\]
where \[\delta(-s) := \frac{1 - \text{ad}s - e^{-\text{ad}s}}{(\text{ad}s)^2}.
\]

Clearly both $\gamma(-s)$ and $\delta(-s)$ are controlled by $s$. Let us show how to bound the term 
\[
\|\omega^{2-\mu}\psi^{2-\nu}(\delta(-s_1)\mathcal D_1 s_1)\gamma(-s_1)\mathcal D_1^{\dagger_{H_0}}s_1 - \delta(-s_2)\mathcal D_1 s_2)\gamma(-s_2)\mathcal D_1^{\dagger_{H_0}}s_2)\|_{0;0,a}
\]
since the other terms are bounded similarly. In order for the following argument to work we need that both $\mu$ and $\nu$ be nonnegative. In fact, $\nu$ must be strictly positive since $v = 0$ corresponds to an indicial root. It is a straightforward computation that the error term belongs to the space $\omega^{-2}\hat C_0^{0,a}(S_+^2)$. This corresponds to $\mu = 0$ and $\nu = 2$. Therefore, since the error term is in the required region we can proceed with the following argument. Using the triangle inequality the term we wish to bound is less than or equal to 
\begin{align*}
&\|\omega^2\delta(-s_1)\mathcal D_1s_1\left( \gamma(-s_1)\mathcal D_1^{H_0^{\dagger}}(s_1 - s_2) + (\gamma(-s_1) - \gamma(-s_2))\mathcal D_1^{H_0^{\dagger}}s_2\right)\|_{0;0,a} + \\
&\|\omega^2\gamma(-s_2)\mathcal D_1^{H_0^{\dagger}}s_2 \Bigl( \delta(-s_2)\mathcal D_1(s_1 - s_2) + (\delta(-s_1) - \delta(-s_2))\mathcal D_1s_1\Bigr) \|_{0;0,a}
\end{align*}

Let us focus on the first line of this bound since the second line can be bounded in a similar way. The condition $\|\psi^{-2}s_1\|_{2;0,a} + \|\psi^{-2}s_1\|_{2;0,a} \le \epsilon$ implies that both $\gamma(-s_i)$ and $\delta(-s_i)$ are bounded in the supremum norm by a constant that only depends on $\epsilon$. Now, 
\begin{align*}
\|\omega\mathcal D_1^{H_0^{\dagger}}(s_1 - s_2)\|_{0;0,a} &\le \|\omega\psi^{-1}\mathcal D_1^{H_0^{\dagger}}(s_1 - s_2)\|_{0;0,a} \\
&\le \|\psi^{-2}(s_1 - s_2)\|_{2;0,a}.
\end{align*}

Similarly $\|\omega\mathcal D_1^{H_0^{\dagger}}s_2\|_{0;0,a} \le \|\psi^{-2}s_2\|_{2;0,a}$ and $\|\omega\mathcal D_1s_1\|_{0;0,a} \le \|\psi^{-2}s_1\|_{2;0,a}$. Finally, 
\begin{align*}
|\gamma(-s_1) - \gamma(-s_2)| &\le C|s_1 - s_2| \\
&\le C\|s_1 - s_2\|_{2;0,a}\\
&\le C\|\psi^{-2}(s_1 - s_2)\|_{2;0,a},
\end{align*}
where the constant $C$ can change from line to line but is dependent only on $\epsilon$. Putting all these bounds together, it is straightforward to obtain the required inequality and therefore we are done.

\section{A priori estimates}
By the results of the previous section, given a solution to the TEBE \eqref{TEBE} of the form \eqref{ansatz} 
for some $\beta_0$, there is an open interval $(\beta_0-\varepsilon, \beta_0 + \varepsilon)$ such that \eqref{TEBE}
admits a solution of this specified form for any $\beta$ in this interval. To prove that solutions exist for every $\beta \in (-\pi/6, \pi/6)$,
or equivalently $\zeta \in [0,1/2)$, we invoke a standard continuity argument. It suffices to show that if $\beta_j$ 
converges to some value in the appropriate open interval, then the solutions $(u_j(\sigma), v_j(\sigma))$ associated
to these parameters have a subsequence which converges to a solution for the limiting value of $\beta$.   Given that this 
is an ODE system, it suffices to deduce uniform $C^1$ bounds for $(u_j, v_j)$ along with bounds which allow us to conclude
that the asymptotic conditions \eqref{asymptotics} are preserved in the limit. This is the goal of this section. 

\subsection{Preliminary bounds and a first integral}
We begin by rewriting \eqref{odesystem} in a slightly nicer form.  Observe that the linear part of these equations can
be rewritten as
\begin{equation*}
(\sigma^2+1)u''+\sigma u' = ((\sigma^2+1)^{1/2}\pa_{\sigma})^2 u.
\end{equation*}
This suggests the change of variables
\begin{equation*}
\frac{\, d\sigma}{\, d\tau} = (\sigma^2+1)^{1/2}, \ \ \mbox{i.e.,}\ \ \sigma = \sinh \tau.
\end{equation*}

Setting $\zeta = \sin \beta$, we also define the quantities
\begin{equation}\label{compactnotation}
V(\tau) :=1- \zeta^2 (k+1)v, \quad T(\tau):= \tanh \tau.
\end{equation}
In terms of these, \eqref{odesystem} now takes the following simpler form
\begin{equation}\label{compactodesystem}
\begin{split}
u'' &= -e^{-2u}(4V^2+\zeta^2v'^2 + 4\zeta^2Vv'T) \\
v'' &= u'(2v' + 4VT) - 4(k+1)V.
\end{split}
\end{equation}
The reader should beware that in \eqref{compactodesystem}, the derivatives $u', v'$, etc.\ are taken with respect to $\tau$,
while in \eqref{odesystem} they are with respect to $\sigma$. Henceforth we only use the variable $\tau$, so hopefully this should
not cause confusion. The function $T(\tau)$ is not intended as a new independent variable. 

We also state the asymptotic conditions: 
\begin{equation}\label{asymp}
\begin{split}
& u(\tau) \sim \log \tau, \ \tau \searrow 0,\ \ u(\tau) \sim (k+1) \tau, \ \tau \nearrow \infty, \\ 
& v(\tau) \sim 2k\tau^2, \tau \searrow 0,\ \ |v(\tau)| \leq C e^{ (k+1)(1-\epsilon)\tau}\ \ \tau \nearrow \infty.
\end{split}
\end{equation}

\begin{lemma}\label{lemma1}
If $(u,v)$ is a solution of \eqref{compactodesystem} and satisfies the asymptotic conditions \eqref{asymp}, then 
\begin{equation}\label{integrated}
u'^2 = (4V^2- \zeta^2v'^2)e^{-2u} + (k+1)^2.
\end{equation}
\end{lemma}
\begin{proof}
Multiply the first equation by $u'$ and the second equation by $\zeta^2v'e^{-2u}$ and add the two equations.  After some
calculaiton this leads to the identity
\begin{equation*}
\begin{split}
u'u'' &= - u'e^{-2u}(4V^2 - \zeta^2v'^2) - \zeta^2e^{-2u}v'(v''+4(k+1)V)\\
&= \frac{1}{2}((4V^2 - \zeta^2v'^2)e^{-2u})'.
\end{split}
\end{equation*}
Integrating this, we obtain the desired formula, at least up to some additive constant.  

To calculate the constant, we integrate from $1$ to $\tau$. This gives
\[
u'(\tau)^2 = \left(4 (1 - \zeta^2 (k+1) v(\tau)^2)^2 - \zeta^2 v'(\tau)^2\right) e^{-2 u(\tau)} + C_1
\]
for every $\tau > 1$, where $C_1$ arises from the boundary terms at $\tau = 1$.    Since 
$v(\tau)$ and $v'(\tau)$ grow at some exponential rate less than $k+1$ as $\tau \to \infty$,
the first quantity on the right converges to $0$.  Since $u'(\tau) \to k+1$, we conclude
that $C_1 = (k+1)^2$, as claimed. 
\end{proof}

We can deduce some immediate useful consequences from this.
\begin{lemma}\label{lemma2}
If $(u,v)$ solves \eqref{compactodesystem}, then $u'$ is monotone decreasing and 
\begin{equation}\label{first bound}
4V^2 - \zeta^2v'^2 \geq 0
\end{equation}
for $\tau > 0$. 
\end{lemma}
\begin{proof}
We first show that $u'' = -e^{-2u}(4V^2+\zeta^2v'^2 + 4\zeta^2Vv'T)  \leq 0$.  To do this, note that
\begin{align*}
4V^2+ \zeta^2 v'^2 +4\zeta^2 Vv'T &\geq 4V^2(1-\zeta^2) + \zeta^2(4V^2 + v'^2 -4|Vv'|) \\
&= 4V^2(1-\zeta^2) + \zeta^2(2|V| - |v'|)^2\\
&\geq 4V^2(1-\zeta^2)
\end{align*}
Since $\zeta < 1/2$, this implies that $u''\le 0$ everywhere, so $u'$ is decreasing to its limit $k+1$ as $\tau$ grows.
The identity \eqref{integrated} then gives that $4V^2 - \zeta^2v'^2 \geq 0$. 
\end{proof}

\begin{corollary}
The component $v(\tau)$ of the solution satisfies the following preliminary uniform a priori bounds:
\begin{equation}\label{crude v bounds}
e^{-2\zeta(k+1)\tau} \le~ V \le e^{2\zeta(k+1)\tau}, \qquad |v'(\tau)| \le 2\zeta^{-1} e^{2\zeta(k+1)\tau}.
\end{equation}
\end{corollary}
\begin{proof}
By the preceding lemma, 
\[
-2 \le \frac{\zeta v'}{1-\zeta^2 (k+1)v} \le 2,
\]
so integrating from $0$ to $\tau$ and using that $V(0) = 1$, we obtain
\[ 
-2\zeta(k+1)\tau \le \log V(\tau) \le 2\zeta(k+1)\tau.
\] 
This gives the first part of \eqref{crude v bounds}. The second part follows from
this together with \eqref{first bound}.
\end{proof}

\subsection{Improved bounds near zero}
The next step is to derive uniform estimates for $u$ near $\tau = 0$. This uses the initial bounds of the last subsection and 
an improved bound for $v'$. 

To simplify notation further, write $b := k+1$ and $f := e^u$. 

\begin{lemma}\label{lemma3}
There exists a positive constant $C$ and an interval $[0,\tau_0]$ both independent of of $\zeta  < 1/2$ such that for any 
solution $(u,v)$ of equations \eqref{compactodesystem}, $u$ satisfies the upper bound
\begin{equation}
u(\tau) \le \log (2\tau) + C\tau \ \ 0 < \tau < \tau_0.
\end{equation}
\end{lemma}
\begin{proof}
By \eqref{integrated} and the first part of \eqref{crude v bounds} we get
\begin{align*}
u'^2 &= (4V^2 - \zeta^2v'^2)e^{-2u} + b^2\\
&\le 4V^2e^{-2u} + b^2\\
&\le 4e^{4\zeta(k+1)\tau}e^{-2u} + b^2.
\end{align*}
Multiplying both sides by $e^{2u}$ and dividing by $4 + b^2f^2$, 
\[
\frac{f'^2}{b^2f^2+4} \le \frac{4e^{4\zeta(k+1)\tau} + b^2f^2}{4+b^2f^2} \le e^{4\zeta(k+1)\tau}
\]
or equivalently 
\begin{equation*}
\frac{f'}{\sqrt{b^2f^2+4}} \le e^{2\zeta(k+1)\tau}.
\end{equation*}
Since $f(0) = 0$, we can integrate this to get
\begin{equation*}
\frac{1}{b}\tanh^{-1}\left(\frac{bf}{\sqrt{b^2f^2+4}}\right) \le \frac{e^{2\zeta(k+1)\tau} - 1}{2\zeta(k+1)},
\end{equation*}
which in turn implies that 
\[
\frac{bf}{\sqrt{b^2f^2+4}} \le \tanh\left(\frac{b(e^{2\zeta(k+1)\tau} - 1)}{2\zeta(k+1)}\right) \le \frac{b(e^{2\zeta(k+1)\tau} - 1)}{2\zeta(k+1)}.
\]
Since $\zeta \le 1$, there exist $C > 0$ and $\tau_0 > 0$ independent of $\zeta$ such that 
\[
\frac{b(e^{2\zeta(k+1)\tau} - 1)}{2\zeta(k+1)} \le b\tau + C\tau^2
\] 
for $\tau < \tau_0$. Hence for such $\tau$, $f \le 2\tau + C'\tau^2$. The statement of the lemma follows immediately from this last inequality. 
\end{proof}

The bound from below for $u$ near $\tau = 0$ is more subtle, and requires better control of $v'$ near $\tau=0$. 
\begin{lemma}\label{lemma4}
With $\tau_0$ as in the previous lemma, there exists $C > 0$ independent of both $\tau \in [0, \tau_0/2]$ and $\zeta$ 
such that 
\begin{equation}\label{upper v' bound}
v'(\tau)\le C\tau.
\end{equation}
\end{lemma}
\begin{proof}
We look more carefully at the equation for $v$: 
\begin{equation*}
v'' = u'(2v' + 4VT) - 4(k+1)V.
\end{equation*}
By \eqref{crude v bounds}, $4(k+1)V \le C'$ for $\tau \leq \tau_0$. By the same bound, we also know that $V\ge 0$. 
Since $T\geq 0$, we obtain 
\[
v'' - 2u'v' \ge -C.
\] 
Multiplying both sides by $e^{-2u}$, this is equivalent to 
\[ 
-Ce^{-2u} \le (v'e^{-2u})'.
\]
On the other hand, we already showed in the proof of \eqref{lemma2} that 
\begin{equation} 
u'' \le -4V^2e^{-2u}(1-\zeta^2) \leq -C e^{-2u}.
\label{oldu}
\end{equation} 
Combining this with the previous inequality, and using that $V \geq C > 0$ for $\tau \leq \tau_0$, we conclude that
\[
Cu'' \le (v'e^{-2u})'.
\] 
Integrating this from any $\tau \leq \tau_0/2$ to $\tau_0$ gives
\[
C(u'(\tau_0) - u'(\tau)) \le v'(\tau_0)e^{-2u(\tau_0)} - v'(\tau)e^{-2u(\tau)}.
\]
Since $u'(\tau_0) > 0$, we can drop it. We also have a uniform bound $|v'|\le C$ for $\tau \leq \tau_0$. Therefore,
\begin{equation}
v'(\tau) \le C(u'(\tau)e^{2u(\tau)} + e^{2(u(\tau) - u(\tau_0))}).
\label{ubv}
\end{equation}
We must show that there is a uniform bound for the right side. 

Let us start by bounding $u'(\tau)e^{2u(\tau)} = f'f$. From the preceding analysis, we have $f\le C\tau$ for $\tau \leq \tau_0$, uniformly. 
It is also straightforward to deduce from the proof of Lemma \eqref{lemma3} that there is some $C>0$ such that $f'\le C$. 
Taken together, this gives a bound of $C \tau$ for the first term on the right.  

The second term is harder to deal with for the following reason. The factor $e^{-2u(\tau_0)}$ might diverge to infinity as $\zeta
\to \zeta_0 < 1/2$. We show below that this is not the case, but unfortunately that proof relies on the current lemma.
We proceed instead in the following way. Since $u'\ge k+1$, $u(\tau)$ increases monotonically, and 
thus $e^{-2u(\tau_0)} \le e^{-2u(\tau)}$ when $\tau\le \tau_0$. Combining this with \eqref{oldu}, we see that
\begin{equation*}
u'' \le -Ce^{-2u(\tau_0)}.
\end{equation*}
Integrating from $\tau_0/2$ to $\tau_0$ and using that $u$ is increasing, we get that if $\tau \leq \tau_0/2$, 
\[
u(\tau) \le u(\tau_0/2) \le -\frac{C}{4}\tau_0^2e^{-2u(\tau_0)} - \frac{\tau_0}{2}u'(\tau_0) + u(\tau_0) \leq 
- C_1 e^{-2u(\tau_0)} + u(\tau_0),
\]
where $C_1 = -\frac14 C \tau^2$. The last inequality holds since $u'(\tau_0) > 0$. 

Finally, 
\[
-C_1 e^{-2u(\tau_0)} \le u(\tau_0)\ \ \mbox{provided}\ \ u(\tau_0) \leq -C.
\]
This is simply because the left hand side decreases exponentially with $u(\tau_0)$ while the right side only decreases linearly. 
In any case, this now gives
\[
u(\tau) \le 2u(\tau_0)
\] 
and hence, using Lemma \eqref{lemma3},
\[
e^{2(u(\tau)-u(\tau_0))} = e^{u(\tau) + (u(\tau) - 2u(\tau_0)} \le e^{u(\tau)} \le C\tau.
\]

This last inequality gives a contradiction if $u_j(\tau_0) \to -\infty$ as $\zeta_j \to \zeta_0$. 
On the other hand, if $u_j(\tau_0)\ge -C$, then 
\[
e^{2(u(\tau)-u(\tau_0))} \le Ce^{2u(\tau)} \le C\tau^2 \leq C' \tau.
\]
Using this in \eqref{ubv} finishes the proof.
\end{proof}

We now prove the lower bound for $v'(\tau)$. 
\begin{lemma}\label{lemma5}
There exists $C > 0$ such that for any $0 \leq \tau \leq \tau_0/2$, 
\begin{equation}\label{lower v' bound}
v'(\tau) \ge -C\tau.
\end{equation}
\end{lemma}

\begin{proof}
We again use the second equation in \eqref{compactodesystem}. 
Since $0\le~V\le C$ for $\tau \leq \tau_0/2$ and $T\le \tau$ for all $\tau$, we see that
\[
v'' - 2u'v' \le C\tau u'.
\] 
Multiplying both sides by $e^{-2u}$ and integrating from $\tau\in [0,\tau_0/2]$ to $\tau_0$, we obtain  
\[ 
v'(\tau_0)e^{-2u(\tau_0)} - v'(\tau)e^{-2u(\tau)} \le -C(\tau_0e^{-2u(\tau_0)} - \tau e^{-2u(\tau)}) + C\int_{\tau}^{\tau_0} e^{-2u(\xi)}\,d\xi
\leq C\int_{\tau}^{\tau_0} e^{-2u(\xi)}\,d\xi;
\] 
the first term on the right has been dropped since it is negative. By \eqref{crude v bounds}, which implies that  $v' \ge -C$ 
and $u'' \le -Ce^{-2u(\tau_0)}$, we obtain
\begin{align*}
-v'(\tau) &\le C\tau + Ce^{2(u(\tau) - u(\tau_0))} -Ce^{2u(\tau)}\int_{\tau}^{\tau_0}u''(\xi)\,d\xi\\
&\le C\tau + Ce^{2(u(\tau) - u(\tau_0))} + Ce^{2u(\tau)}(u'(\tau) - u'(\tau_0))\\
&\le C\tau + Ce^{2(u(\tau) - u(\tau_0))} +  Ce^{2u(\tau)}u'(\tau)\\
&\le C\tau + Ce^{2(u(\tau) - u(\tau_0))}\\
&\le C\tau,
\end{align*}
since we showed earlier that $Ce^{2u(\tau)}u'(\tau)\le C\tau$. From the proof of Lemma \eqref{lemma4} we can also bound 
$ Ce^{2(u(\tau) - u(\tau_0))}$ by $C\tau$. This gives the desired inequality. 
\end{proof}

We have now established, in \eqref{upper v' bound} and \eqref{lower v' bound}, the uniform estimate $|v'(\tau)| \leq C\tau$ 
for $0 \leq \tau \leq \tau_0$, where $C$ and $\tau_0$ are independent of $\zeta$.   The next goal is to obtain
uniform bounds for $u$. 
\begin{lemma}\label{lemma6}
There exist constants $C > 0$ and $\tau_0 > 0$, both independent of $\zeta$, such that if $(u,v)$ solves \eqref{compactodesystem} 
and satisfies the specified asymptotics, then 
\begin{equation}
|u(\tau) - \log (2\tau)| \leq C \tau^2, \qquad |u'(\tau) - \tau^{-1}| \leq C \tau
\end{equation}
for $\tau \leq \tau_0$. 
\end{lemma}
\begin{proof}
As before, we establish upper and lower bounds separately. The proof is similar to that of Lemma \eqref{lemma3}. 

Starting from \eqref{integrated}, and using $\zeta^2v'^2 \le C\tau^2$ for $\tau \leq \tau_0$, we also get
$V\ge 4 -C\tau^2$ on this same  interval. This gives
\[
u'^2 \ge (4 - C\tau^2)e^{-2u} +b^2,
\] 
so that
\[
\frac{f'^2}{b^2f^2+4} \ge \frac{4-C\tau^2 + b^2f^2}{b^2f^2+4} \ge 1-C\tau^2, \ \ \mbox{or more simply}\ \ \ 
\frac{f'}{\sqrt{b^2f^2+4}} \ge 1-C\tau^2,
\] 
since $f' = u' e^u > 0$. Integrating from $0$ to $\tau$, noting that $f(0) = 0$, we obtain 
\[
\frac{bf}{\sqrt{b^2f^2+4}} \ge \tanh(b\tau - C\tau^3) \ge b\tau - C'\tau^3.
\]
We can rewrite this as 
\[
bf \ge \frac{ b\tau - C\tau^3}{\sqrt{1-( b\tau - C\tau^3)^2}} \geq 2b \tau - C \tau^2,
\] 
whence the desired inequality. 

For the opposite inequality, $|v'(\tau)| \leq C\tau^2$ implies that $V\le 4+ C\tau^2$, so repeating the steps above, we get
\[
\frac{bf}{\sqrt{b^2f^2+4}} \le \tanh(b\tau + C\tau^3) \le b\tau + C\tau^3,
\]
from which the required inequality is straightforward. 

For the bounds on $u' - \tau^{-1}$, we prove only the lower bound since the upper bound is similar. We combine
\[
f' \ge (1- C\tau^2)\sqrt{b^2f^2+4} > 2(1 - C\tau^2), \ \ \mbox{and}\ \ \  f'= u'f \le u'(2\tau + C\tau^3)
\]
to get
\[ 
u' \ge \frac{2 - C\tau^2}{2\tau + C\tau^3} \ge \frac{1}{\tau} - C\tau.
\]
\end{proof}

\subsection{Bounds away from zero}
We have now proved bounds on the components $u$, $v$ in some uniform interval $[0, \tau_0]$.  It remains to establish appropriate uniform 
bounds on $[\tau_0, \infty)$. 

As a first step, decompose $u$ as $u = u_0 + w$, where $u_0$ is the model charge $k$ knot solution for the group $G = \mbox{SU(2)}$ when $\zeta = 0$.
Explicitly, 
\begin{multline*}
u_0(\sigma) = \log\left(\frac{((\sigma^2+1)^{1/2} + \sigma)^{k+1} - ((\sigma^2+1)^{1/2} - \sigma)^{k+1}}{2(k+1) } \right) \\
\Rightarrow u_0(\tau) =  \log \left(  ( e^{(k+1)\tau} - e^{-(k+1)\tau} )/2(k+1)\right) = \log \sinh( (k+1)\tau) - \log(k+1).
\end{multline*}
This function has precisely the same leading asymptotic  behavior as the solution $u$,  both as $\tau \to 0$ and as $\tau \to \infty$. 

\begin{lemma}\label{lemma7}
Choosing $\tau_0 > 0$ as in the previous subsection, if $\zeta < 1/2$, so $d_{\zeta} := 1 - 2\zeta > 0$, then
\begin{align}
-C \le  & \, w \le C_{\zeta}\label{w1} \\
-Ce^{-2(k+1)\tau} \le & \, w' \le Ce^{-d_{\zeta}(k+1)\tau}\label{w2}.
\end{align}
\end{lemma}
\begin{proof}
Given that $e^{-2u_0} \le Ce^{-2(k+1)\tau}$, where $C$ depends only on $\tau_0$, then using \eqref{crude v bounds}, we see that
\begin{align*}
(u_0' + w')^2 &= (4V^2 - \zeta^2v'^2)e^{-2u_0}e^{-2w} + b^2\\
&\le 4V^2e^{-2u_0}e^{-2w} +b^2\\
&\le Ce^{ 2 (2\zeta-1)(k+1)\tau}e^{-2w} + b^2 \le ( Ce^{-d_{\zeta}b\tau}e^{-w} + b)^2.
\end{align*}
Since $b  \leq u' = u_0' + w'$ and $b \leq u_0'$ as well, we have
\begin{equation}
b - u'_0 \le w' \le Ce^{-d_{\zeta} b\tau}e^{-w}.
\label{bdw}
\end{equation}

We analyze the left inequality first.  We know that 
\[
(u_0')^2 = 4e^{-2u_0} + b^2 \Rightarrow (u_0')^2 - b^2 = 4e^{-2u_0}
\]
hence 
\[ 
w' \geq b - u_0' \geq -Ce^{-2b\tau}
\] 
for $\tau \geq \tau_0$. The constant $C$ is uniform in $\zeta$ since $e^{-2u_0}$ is independent of $\zeta$. 
We have established uniform control of $w(\tau_0)$ in \eqref{lemma6}, so integrating from $\tau_0$ to $\tau$ proves the lower bounds
for both $w$ and $w'$. 

As for the upper bounds, multiply both sides of the right inequality in \eqref{bdw} by $e^w$ to get 
\[ 
(e^{w})' \le Ce^{-d_{\zeta}b\tau}, \ \ \ \tau \geq \tau_0.
\]
Integrating from $\tau_0$ to $\tau$ yields 
\[ 
w(\tau) - w(\tau_0) \le \frac{C}{d_{\zeta}b}e^{-d_{\zeta}b\tau}.
\] 
Since, as above, we have a bound on $w(\tau_0)$, we get the upper bound in \eqref{w1}.   The lower bound in \eqref{w1}, which
is already established, gives $e^{-w} \le C$, and hence 
\[ 
w' \le Ce^{-d_{\zeta} b\tau}e^{-w} \le Ce^{-d_{\zeta} b\tau}.
\]
This concludes the proof.
\end{proof}

We finally prove global bounds for $V$ and $v'$: 
\begin{lemma}\label{lemma8}
There exists $C = C(\zeta_0) > 0$ such that for $\zeta \leq \zeta_0 < 1/2$ and $\tau \geq \tau_0$, 
\[
 |V| \leq C, \quad -Ce^{-2\tau} \le v' \le C e^{-d_{\zeta}b\tau}.
\]
\end{lemma}
\begin{proof}
Rerite the second equation in \eqref{compactodesystem} as 
\[
(v'e^{-2u})' = 4e^{-2u}V(u'T - (k+1)).
\]
By \eqref{w2}, we have 
\[ 
b \le u' \le b + Ce^{-d_{\zeta}b\tau}.
\] 
Combining this with 
\[ 
1 - 2e^{-2\tau} \le T \le 1,
\]
leads to 
\begin{equation}\label{bootstrap}
-Ce^{-2u-2\tau}V \le (v'e^{-2u})' \le Ce^{-2u-d_{\zeta}b\tau}V, \ \ \ \tau \geq \tau_0/
\end{equation}
Then \eqref{crude v bounds} and the bounds for $u$ give
\begin{equation*}
-Ce^{-((2-2\zeta)b+2)\tau} \le (v'e^{-2u})' \le Ce^{-(d_{\zeta}+2-2\zeta)b\tau}.
\end{equation*}
Integrating from $\tau \geq \tau_0$ to $\infty$ and using \eqref{crude v bounds} to get the exponential decay of $v'e^{-2u}$,
we arrive at
\begin{equation*}
-Ce^{(2\zeta b - 2)\tau} \le v' \le Ce^{(2\zeta - d_{\zeta})b\tau}.
\end{equation*}

This is an improvement compared to the second bound in \eqref{crude v bounds}: the exponent in the lower bound has dropped from 
$2\zeta b$ to $2\zeta b - 2$, while on the right hand side it has dropped by $d_{\zeta}b$.  
Repeating this argument iteratively, then so long as $d_\zeta > 0$, we eventually obtain exponentially decaying functions
both above and below, as desired.  Consider only the upper bound for simplicity. Let $n_{\zeta}$ be the largest integer 
for which $2\zeta - n_{\zeta}d_{\zeta} > 0$.  Integrating from $\tau_0$ to $\tau$ and using \eqref{upper v' bound} we get that 
\[
V(\tau) \le C + \frac{C}{(2\zeta - d_{\zeta})b}e^{(2\zeta - d_{\zeta})b\tau} \le \frac{C}{(2\zeta - d_{\zeta})b}e^{(2\zeta - d_{\zeta})b\tau}
\]
for some possibly different constant $C$. Plugging this improvement into \eqref{bootstrap} and repeating the steps,
we obtain inductively 
\[
V(\tau) \le C e^{(2\zeta - n_{\zeta}d_{\zeta})b\tau}\big/ \prod \limits_{n = 1}^{n_{\zeta}}(2\zeta - nd_{\zeta})b) .
\]

We must distinguish between two cases. If $2\zeta - (n_{\zeta}+1)d_{\zeta} < 0$, then using this final estimate in \eqref{bootstrap}  gives
the exponentially decaying upper bound for $v'$.  Repeating the steps above, we also see that $V \leq C$. Using this bound \eqref{bootstrap} 
one last time, we obtain the required upper bound for $v'$. 

The other case is when $2\zeta - (n_{\zeta}+1)d_{\zeta} = 0$. We then must repeat the steps one additional time; at the first pass, we deduce
that $v' \leq C$, so $V \leq C \tau$, but this implies that $v'$ is bounded above by an exponentially decreasing function, and hence that
$V \leq C$.
\end{proof}

\subsection{Conclusion of the continuity argument}
Now that have established uniform bounds for $u$ and $v$ for $\zeta$ in any interval $[0, \zeta_0]$, $\zeta_0 < 1/2$, we can complete
the proof of existence of model knot solutions for every $\zeta \in [0, 1/2)$. 

\begin{theorem}
Suppose that $\zeta_j \nearrow \bar{\zeta} < 1/2$ and that $(u_j,v_j)$ is a corresponding sequence of solutions to \eqref{compactodesystem} 
with parameter $\zeta_j$. Then there exists a subsequence $(u_{j'}, v_{j'})$ which converges smoothly on every compact interval $0 < a \leq b < \infty$,
and such that the limiting solution $(\bar{u}, \bar{v})$ satisfies the same asymptotics as $\tau \to 0$ and $\tau \to \infty$. 
\end{theorem}
\begin{proof}
Write $u_j = u_0 + w_j$ as before. By Lemmas \eqref{lemma6} and \eqref{lemma7}, $w_j$ is uniformly bounded on $\mathbb R^+$ and $|w_j'(\tau)|$
decays at a uniform exponential rate as $\tau \to \infty$ and is uniformly bounded near $\tau = 0$. Lemmas \eqref{lemma4} , \eqref{lemma5} 
and \eqref{lemma8} provide the analogous uniform bounds for $v_j$ and $v_j'$.  The equation itself now gives uniform bounds on 
$u_j''$ and $v_j''$ on any compact interval, and this may be bootstrapped to bounds for any derivative.  By taking a diagonal subsequence
over increasing intervals, we obtain limiting functions $(\bar{u}, \bar{v})$, and by our uniform bounds near $0$ and $\infty$, this
limiting solution also has the same asymptotics. 

Of course, $(\bar{u}, \bar{v})$ is real analytic as well. Indeed, rewrite \eqref{compactodesystem} as 
\begin{equation}\label{wsystem}
\begin{split}
w'' &= - e^{-2u_0}(e^{-2w}(4V^2+\zeta^2v'^2 + 4\zeta^2Vv'T) -4)\\
v'' &= (u_0' + w')(2v' + 4VT) - 4(k+1)V.
\end{split}
\end{equation}
This system can be written in the form
\[ 
\tau U' = F( U, \tau),
\] 
where $\vec U:= (u,u',v,v')$ and $F$ is analytic in its arguments,   This is a regular singular equation, and classical theory gives a
convergent series expansion. Analyticity near $\tau = 0$ then follows since the indicial roots of this system at $\tau = 0$ 
are positive integers. 
\end{proof}

\printbibliography

\end{document}